\documentclass[11pt]{article}
\usepackage[margin=0.8in]{geometry}
\usepackage{libertine}

\usepackage{amsfonts}
\usepackage{bbm}
\usepackage{xcolor}

\usepackage{amsthm,thmtools,thm-restate}
\usepackage{mathtools} 
\usepackage{amsmath} 
\usepackage{algorithm}
\usepackage{algorithmic}
\usepackage{tabularx}
\usepackage{cleveref}

\usepackage{enumitem}

\crefname{ALC@line}{line}{lines} 

\DeclareMathOperator*{\argmax}{arg\,max}
\DeclareMathOperator*{\argmin}{arg\,min}

\newtheorem{theorem}{Theorem}
\newtheorem{remark}{Remark}
\newtheorem{corollary}[theorem]{Corollary}
\newtheorem{definition}{Definition}
\newtheorem{lemma}[theorem]{Lemma}

\newtheorem{claim}{Claim}

\usepackage[square,semicolon,numbers]{natbib}
\newcommand\numberthis{\addtocounter{equation}{1}\tag{\theequation}}

\newcommand{\alloc}{\mathcal{A}}
\newcommand{\len}{\mathrm{len}}
\newcommand{\EFt}{\mathrm{EFE}3}
\newcommand{\EFk}{\mathrm{EFE}k}
\newcommand{\EQt}{\mathrm{EQE}3}
\newcommand{\EQk}{\mathrm{EQE}k}
\newcommand{\EF}{\mathrm{EF}}
\newcommand{\EQ}{\mathrm{EQ}}
\newcommand{\EFo}{\mathrm{EF}1}
\newcommand{\EQo}{\mathrm{EQ}1}

\newcommand{\calI}{\mathcal{I}}
\newcommand{\calG}{\mathcal{G}}
\newcommand{\calA}{\mathcal{A}}
\newcommand{\calD}{\mathcal{D}}
\newcommand{\calC}{\mathcal{C}}

\newcommand{\calF}{\mathcal{F}}
\usepackage{graphicx} 

\title{\bfseries Fair Division Beyond Monotone Valuations with Applications to \\ Equitable Graph Partitioning}

\author{Siddharth Barman\thanks{Indian Institute of Science. {\tt barman@iisc.ac.in}} \and Paritosh Verma\thanks{Purdue University. {\tt paritoshverma97@gmail.com}}}

\date{}

\begin{document}

\maketitle

\begin{abstract}
This paper studies fair division of divisible and indivisible items among agents whose cardinal preferences are not necessarily monotone. We establish the existence of fair divisions and develop approximation algorithms to compute them. We address two complementary valuation classes, subadditive and nonnegative, which go beyond monotone functions. Considering both the division of cake (divisible resources) and allocation of indivisible items, we obtain fairness guarantees in terms of (approximate) envy-freeness ($\mathrm{EF}$) and equability ($\mathrm{EQ}$). 
 
In the context of envy-freeness, we prove that an $\mathrm{EF}$ division of a cake always exists under cake valuations that are subadditive and globally nonnegative (i.e., the value of the entire cake for every agent is nonnegative, but parts of the cake can be burnt). This result notably complements the nonexistence of $\mathrm{EF}$ allocations for burnt cakes known for more general valuations. For envy-freeness in the indivisible-items setting, we establish the existence of $\mathrm{EFE}3$ allocations for subadditive and globally nonnegative valuations; again, such valuations can be non-monotone and can impart negative value to specific item subsets. In addition, we obtain universal existence of $\mathrm{EFE}3$ allocations under nonnegative valuations. 

We study equitability under nonnegative valuations. Here, we prove that $\mathrm{EQE}3$ allocations always exist when the agents' valuations are nonnegative (and possibly non-monotone). Also, in the indivisible-items setting, we develop an approximation algorithm that, for given nonnegative valuations, finds allocations that are equitable within additive margins.   

Our results have combinatorial implications. For instance, the developed results imply the following novel results: (i) The universal existence of proximately-dense subgraphs: Given any graph $G=(V, E)$ and integer $k$ (at most $|V|$), there always exists a partition $V_1, V_2, \ldots, V_k$ of the vertex set such that the edge densities within the parts, $V_i$, are additively within four of each other, and (ii) The universal existence of equitable graph cuts: Given any graph $G=(V,E)$ and integer $k$ (at most $|V|$), there always exists a partition $V_1, V_2, \ldots, V_k \neq \emptyset$ of the vertex set such that the cut function values of the parts, $V_i$, are additively within $5 \Delta +1$ of each other; here, $\Delta$ is the maximum degree of $G$. Further, such partitions can be computed efficiently. In addition to being interesting in and of itself, this result highlights the reach of the developed guarantees beyond fair division and even algorithmic game theory. 
\end{abstract}
\thispagestyle{empty}
\newpage

\tableofcontents
\thispagestyle{empty}
\newpage

\setcounter{page}{1}
\section{Introduction}
Fairly dividing resources and tasks among agents with individual preferences is a ubiquitous problem extensively studied in mathematical economics and computer science. A vast majority of this literature addresses models wherein the agents' preferences are monotone. That is, most prior works here assume that a clear distinction exists as to whether the underlying items are goods, with monotonically increasing valuations, or chores, with monotone disutilities. 

However, the monotonicity assumption is inapplicable in various allocation domains: \emph{``Anything in excess is a poison''} (Theodore Levitt). Indeed, superfluous resources can drive down valuations in settings where free disposal is infeasible.\footnote{See, e.g., {\tt https://www.bbc.com/news/business-52350082}.} This negative-marginals property holds even in some standard fair division contexts, such as course allocation and job assignments, where number of courses, or workload, beyond a threshold can be detrimental. Another example that highlights the relevance of non-monotone valuations comes from partnership dissolutions, wherein  assets along with liabilities need to be fairly divided. 
The fair allocation of mixed manna \cite{bogomolnaia2017competitive,aziz2018fair} (i.e., the joint allocation of positively-valued goods and negatively-valued chores) also entails non-monotone valuations. 
 
Even when considering mathematical models, non-monotonicity is well-motivated and arises naturally. For instance, one can start with a monotone increasing valuation function, $v$, and a monotone, additive cost function, $c$. Still, the induced quasilinear utility---obtained by subtracting $c$ from $v$---can be non-monotone. 

Furthermore, many of the existential and algorithmic guarantees developed for fair division are sensitive to whether the agents' valuations are positive, negative, or both. A substantiating example here is the dichotomy known for fair cake division. A cake represents a heterogenous, divisible resource, and the seminal work of Su~\cite{Su1999rental} shows that, under mild conditions and nonnegative valuations, an envy-free (fair) division always exists. By contrast, when the valuations can be negative (parts of the cake are burnt) and may depend on the entire division, envy-free allocations exist iff the number of agents is a prime power~\cite{avvakumov2021envy}. 

Another technical challenge that marks a separation between monotone increasing and monotone decreasing valuations comes from the discrete fair division context: For indivisible goods (i.e., non-negatively valued items) and under additive valuations, there always exists an allocation that is both envy-free up to one good ($\EFo$) and  Pareto efficient (PO) \cite{caragiannis2019unreasonable}. However, the existence of an $\EFo$ and PO allocation for indivisible chores---i.e., items with additive and negative valuations---stands as an intriguing open question. These examples highlight the potential and significance of studying non-monotone valuations. Additionally, essentially all the known techniques used in the monotone context (envy-cycle elimination, market equilibrium, etc.) fail for non-monotone valuations, making it further challenging. 
 
Motivated by such considerations, the current work aims to extend existential and algorithmic fair division guarantees beyond monotone valuations. We address two complementary valuation classes, subadditive and nonnegative. Considering both the division of cake and allocation of indivisible items, we obtain fairness guarantees in terms of (approximate) envy-freeness and equability. We describe these fair division settings next.

As mentioned, the cake represents a divisible resource and is modeled as the interval $[0,1]$. Fair cake division has been studied for more than half a century~\cite{steinhaus1948problem, dubins1961cut}. The cardinal preferences of the $n$ agents over the cake are specified via valuations $f_1, \ldots, f_n$; here, $f_i(I) \in \mathbb{R}$ denotes the value that agent $i$ has for any interval $I \subseteq [0,1]$. The paper studies fair cake division under a typical requirement that every agent receives a connected piece (i.e., an interval) of the cake. Hence, the objective here is to partition the cake [$0, 1]$ into exactly $n$ pairwise disjoint intervals and assign them among the $n$ agents. Such cake partitions as referred to as contiguous cake divisions $\calI = (I_1, \ldots, I_n)$, wherein $I_i$ denotes the interval assigned to agent $i$.   

We also address fair division of indivisible items. This thread of research has received significant attention in recent years \cite{amanatidis2023fair}. The valuations in this context are set functions $v_1, \ldots, v_n$. Specifically, for any subset $S$ of the $m$ indivisible items, $v_i(S)$ denotes the value that agent $i$ has for receiving bundle $S$. To distinguish that $v_i$s are not necessarily monotone, we will use the term (indivisible) items and not goods or chores. In the indivisible-items setting, the goal is to partition the set of items $[m]$ into pairwise disjoint subsets and assign them among the agents. The term allocation $\calA = (A_1, \ldots, A_n)$ denotes such an $n$-partition of $[m]$, in which $A_i$ is the subset of items assigned to agent $i$.  \\

The work obtains fairness guarantees in terms of two central notions of fairness, envy-freeness \cite{foley1966resource, stromquist1980cut} and equitability \cite{dubins1961cut}. 

\vspace*{5pt}

\noindent
{\bf Envy-Freeness.} A division is said to be envy-free ($\EF$) if each agent prefers the bundle assigned to her over that of anyone else. Hence, a contiguous cake division $\calI=(I_1, \ldots, I_n)$ upholds the envy-freeness criterion if $f_i(I_i) \geq f_i(I_j)$ for all agents $i$ and $j$. As mentioned, under mild conditions and for nonnegative valuations, in particular, an envy-free, contiguous cake division always exists~\cite{Su1999rental}. However, this positive result does not extend to negative and non-monotone valuations and does not admit an algorithmic counterpart either. When the agents' valuations $f_i$ are negative---the cake is burnt---an envy-free division is not guaranteed to exist \cite{avvakumov2021envy}. Further, a contiguous and exact envy-free division cannot be computed in finite time, even for nonnegative valuations that are additive \cite{stromquist2008envy}.\footnote{While one obtains a finite-time algorithm by relinquishing the contiguity requirement (i.e., for the case in which each agent can receive a finite union of intervals), the complexity of the known algorithm is hyper-exponential \cite{aziz2016discrete}.} These negative results substantiate the study of nonnegative/non-monotone valuations and approximation algorithms for cake division.  

In the context of indivisible items, envy-freeness is not a feasible criterion; allocating a single indivisible item among agents that positively value it will always result in envy. Hence, relaxations of envy-freeness have been extensively studied in discrete fair division. Of particular note here is the notion of envy-freeness up to one good ($\EFo$) \cite{lipton2004approximately,budish2011combinatorial}. Under this criterion, an allocation $\calA=(A_1, \ldots, A_n)$ of the indivisible items among the agents is deemed to be fair if, for all agents $i,j$, any existing envy can be mitigated by notionally shifting one item from $A_i$ or $A_j$. Specifically, when all the items are goods (i.e., have non-negative marginal values), we have the following version of the definition: An allocation $\calA=(A_1, \ldots, A_n)$ is said to be $\EFo$ if, for all agents $i,j$ there exists a good $g \in A_j$ such that $v_i(A_i) \geq v_i(A_j \setminus \{g\})$. In the case of all chores (i.e., when all the items have non-positive marginal values), an allocation $\calA=(A_1, \ldots, A_n)$ is said to be $\EFo$ if, for all agents $i,j$ there exists a chore $c \in A_i$ such that $v_i(A_i \setminus \{c\}) \geq v_i(A_j)$. $\EFo$ allocations are known to exist in the two non-exhaustive settings wherein either all items are goods \cite{lipton2004approximately}, or all are chores \cite{bhaskar2021approximate}. $\EFo$ allocation also exist in mixed setting of goods and chores (mixed manna), though this result holds specifically for additive valuations \cite{aziz2022fair}.  

The current work goes beyond monotone, increasing or decreasing, valuations -- in this paper, the marginal value $v(S \cup \{j \}) - v(S)$ of an item $j$ can be positive or negative, depending on the underlying subset $S$. Here, we obtain results for a further relaxation, namely envy-freeness up to three edits ($\EFt$). We say that an allocation $\alloc = (A_1, \ldots, A_n)$ is $\EFt$ if the envy that any agent $i$ has towards any other agent $j$ can be resolved by shifting at most three items in or out of $A_i$ and $A_j$ combined (see Definition \ref{definition:EFk}). 

Notably, for the valuation classes considered in this paper, nontrivial fairness guarantees (e.g., $\EFk$ for any constant $k$) were not known prior to the current work. In fact, on the non-monotone front in discrete fair division, fairness guarantees were known primarily for additive valuations \cite{aziz2018fair}. Extending this important frontier, this work establishes that relevant fairness guarantees are feasible even under valuations that are non-monotone and not necessarily additive. 

\vspace*{5pt}

\noindent {\bf Equitability.}  Equitability requires that all the agents value their individual bundles equally~\cite{dubins1961cut}. For a cake division $\calI = (I_1, \ldots, I_n)$ this corresponds to the equality $f_i(I_i) = f_j(I_j)$, for all agents $i$ and $j$. Equitable cake divisions are known to exist for continuous and nonnegative valuations~\cite{aumann2015efficiency, dobovs2013existence, cheze2017existence, avvakumov2023equipartition}. This existential guarantee extends to other related valuation classes \cite{bhaskar2025connected}. On the algorithmic side, \cite{cechlarova2012computability} provides an efficient algorithm for computing approximately equitable cake divisions under additive valuations. The current work goes past additive valuations and develops an approximation algorithm for nonnegative valuations.  

In the indivisible-items setting, an equitable allocation might not exist, motivating the need for considering relaxations. Towards this, equitability up to one good ($\EQo$) provides a compelling notion \cite{gourves2014near,freeman2019equitable}. $\EQo$ requires that, in the considered allocation $\calA =(A_1, \ldots, A_n)$, existing inequity, between any pair of agents $i,j$, can be reversed by (notionally) shifting some item from $A_i$ or $A_j$. $\EQo$ allocations are known to exist under monotone increasing valuations \cite{barman2024nearly}.\footnote{The result in \cite{barman2024nearly} holds for a stronger criterion, $\EQ$X.} We move past monotone valuations and obtain positive results for equitability up to three edits ($\EQt$). In particular, an allocation $\calA =(A_1, \ldots, A_n)$ is said to be $\EQt$ if, existing inequity, between any pair of agents $i,j$, can be reversed by (notionally) shifting three items in or out of $A_i$ and $A_j$ combined; see Section \ref{section:prelims} for formal definitions. \\

As mentioned, we address two complementary valuation classes, subadditive and nonnegative. Also, we will, throughout, focus solely on valuations that value the empty bundle at zero, i.e., to normalized valuations. In the cake context this corresponds to $f_i(\emptyset) = 0$, for each cake valuation. In the indivisible-items settings, we have $v_i(\emptyset) = 0$, for all the valuations (set functions).  

\vspace*{5pt}

\noindent {\bf Nonnegative Valuations.} For the cake division setting, this class includes valuations $f$ that satisfy $f(I) \geq 0$, for all intervals $I$. Analogously, in the indivisible-items setting, a valuation $v$ is said to be nonnegative if $v(S) \geq 0$, for all item subsets $S$. Indeed, nonnegative valuations are not confined to be monotone. 

\vspace*{5pt}

\noindent {\bf Subadditive Valuations.} These valuations provide a generalization of additivity. Recall that a cake valuation $g$ is $\sigma$-additive if for any two disjoint intervals $I, J \subseteq [0,1]$, we have $g(I \cup J) = g(I) + g(J)$. Extending this definition to incorporate complement-freeness, one obtains $\sigma$-subadditivity: a cake valuation $f$ is said to be $\sigma$-subadditive if  $f(I \cup J) \leq f(I) + f(J)$ for any two {\it disjoint} intervals $I, J \subseteq [0,1]$. Note that this definition (as in the case of additive valuations) considers disjoint intervals. Also, a $\sigma$-subadditive cake valuation $f$ can be negative and non-monotone. 

In the discrete setting, we say that a valuation $v$ is $\sigma$-subadditive if, for any two {\it disjoint} subsets $S$ and $T$, the value of their union is at most the sum of their values, $v(S \cup T) \leq v(S) + v(T)$. Note that, here, we only consider disjoint subsets. Since we mandate the inequality for disjoint subsets and not for all subset pairs, we obtain a strengthening of the form of subadditivity studied in, say, welfare-maximization contexts \cite{nisan2009algorithmic}.\footnote{In fact, requiring the complement-freeness to hold for all subset pairs forces a valuation $h$ to be nonnegative: if we have $h(S \cup T) \leq h(S) + h(T)$ for {\it all} subsets $S$ and $T$, then, in particular for $X \subset Y$, we obtain $h(X \cup Y) = h(Y) \leq h(X) + h(Y)$; this implies $h(X) \geq 0$.} The generalized form of subadditivity considered in this work (i.e., $\sigma$-subadditivity) admits valuations that can be negative and non-monotone; negative valuations imply that the items can have mixed manna. 

We develop positive results for $\sigma$-subadditive valuations that assign a nonnegative value to the grand bundle, i.e., subadditive valuations that are globally nonnegative. For cake valuations $f_i$, this corresponds to $f_i([0,1]) \geq 0$, and for the indivisible-items setting we have $v_i([m]) \geq 0$.  \\

\noindent {\bf Our Results.} We next summarize the contributions of this work; see also Figure \ref{fig:organize}. \\

\begin{figure}[h]
\begin{center}
\includegraphics[scale=0.7]{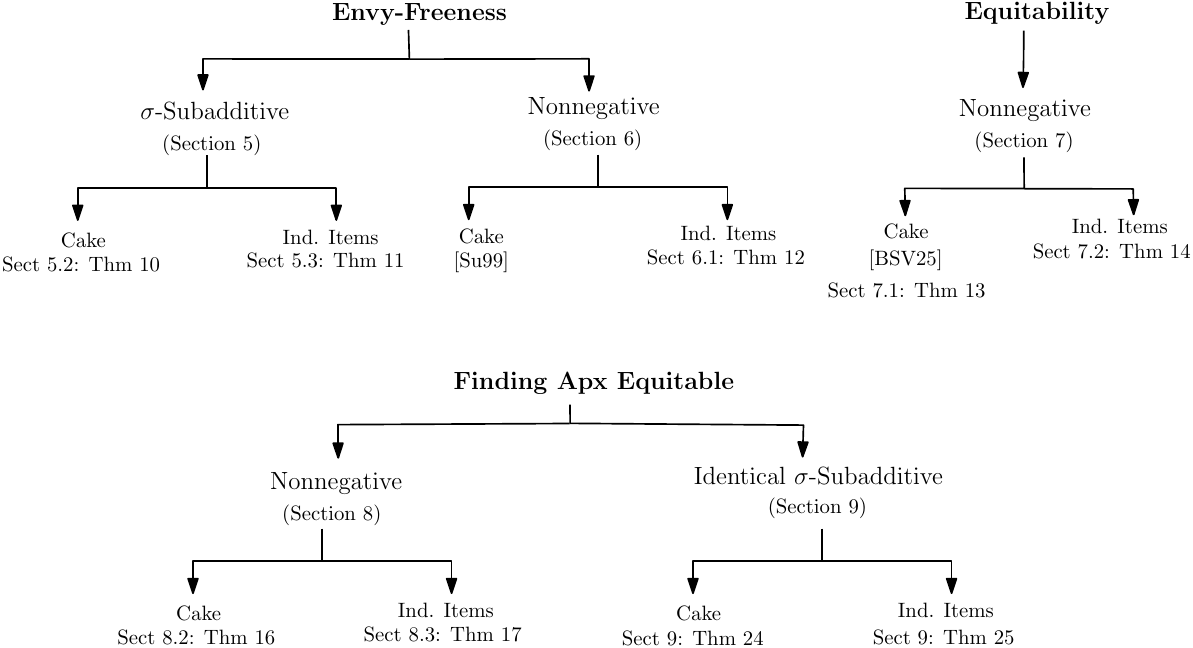}
	\end{center}
	\caption{The figure illustrates the fair division settings considered in this work and the associated theorems.}
\label{fig:organize}
\vspace*{-10pt}
\end{figure}

\noindent {\bf Existence of Envy-Free Divisions.} 
First, we consider subadditive valuations. 
\begin{itemize}[leftmargin = 3mm]
\item We establish universal existence of envy-free cake divisions under cake valuations that are $\sigma$-subadditive, globally nonnegative, and continuous (\Cref{theorem:burnt-cake-envy-free}). This result cover partially-burnt cakes and thus contrasts with the nonexistence of envy-free divisions for burnt cakes~\cite{avvakumov2021envy}.
\item  For indivisible items, we prove that, under $\sigma$-subadditive and globally nonnegative valuations, $\EFt$ allocations always exist (\Cref{theorem:ef-three-subadditive}). Notably, this result is proved by establishing that the multilinear extension of arbitrary subadditive functions also exhibits subadditivity (\Cref{lemma:subadditive-mle}); to the best of our knowledge, this connection between subadditivity and multilinear extensions is novel and might be of independent interest.
\end{itemize}

\noindent Then, for envy-freeness under nonnegative valuations, we establish the following result: $\EFt$ allocations of indivisible items always exist under nonnegative valuations (\Cref{theorem:ef-three-nonnegative}). This existential guarantee is obtained by connecting with the classic envy-free cake division result of Su~\cite{Su1999rental}. \\

\noindent {\bf Existence of Equitable Divisions.}
In the context of equitability, we obtain results for nonnegative valuations. In particular, we prove the following existential guarantees.
\begin{itemize}[leftmargin = 3mm]
\item For cake division, it is known that, under nonnegative and continuous valuations, there always exists an equitable cake division that additionally upholds an ordering property \cite{bhaskar2025connected}. We provide a complementary proof of this result via a fixed-point argument (\Cref{theorem:eq-cake-nn}). 
\item We further prove that, discrete fair division instances with nonnegative valuations always admit an $\EQt$ allocation (\Cref{theorem:eq-three-nonnegative}).  \\
\end{itemize}

\noindent {\bf Finding Approximately Equitable Divisions.} Finally, we develop algorithms for finding allocations that are equitable within additive margins. These algorithms leverage the above-mentioned ordering property.  
The algorithms for nonnegative valuations are as follows. 
\begin{itemize}[leftmargin = 3mm]
\item We prove that, given nonnegative and $\gamma$-Lipschitz continuous cake valuations along with parameter $\varepsilon >0$, we can find a contiguous cake division in which the inequity is additively bounded by $\varepsilon$. The time complexity of the algorithm is polynomial in $n$, $\gamma$, and $1/\varepsilon$ (\Cref{theorem:fptas-for-equitable-cake-div}). To the best of our knowledge this is the first algorithmic result for computing nearly equitable divisions for nonadditive valuations.
\item For indivisible items with nonnegative valuations, we develop an algorithm that finds an allocation in which the inequity is additively bounded by $5 \Lambda + 1$; here, $\Lambda$ is a the maximum marginal value of the valuations (\Cref{theorem:EQ3ApxComputation}). The algorithm also ensures that the returned bundles are nonempty.\footnote{This property has important implications for our graph-theoretic applications.} Our algorithm runs in time polynomial in $n$, $m$, and $\Lambda$. 
\end{itemize}

We also obtain similar algorithmic results for approximate equitability among agents whose valuations are subadditive and identical (Theorems \ref{theorem:cake-alg-id-subadd} and \ref{theorem:EQ3ApxComputation-subadditive}). \\

Our results have multiple interesting implications which includes novel results in graph partitioning; see Section \ref{section:applications} below. While being interesting in and of themselves, they highlight the reach of the developed guarantees beyond fair division. 
\section{Applications Beyond Fair Division}\label{section:applications}
This section presents instantiations of our fairness guarantees. 
The two applications detailed next address equitability with respect to subgraph densities and cut functions, respectively. 

To the best of our knowledge, these graph-theoretic results are novel. We also note that applying fair division guarantees to classic graph-theoretic constructs (cuts and densities) is an interesting stride in itself. This aspect is further substantiated by the fact that we obtain these combinatorial results without using typical techniques,   such as the probabilistic method and local search. In particular, it is unclear if the probabilistic method (because of edge dependencies) can lead to bounds as sharp as the ones obtained in Theorems \ref{application:density-part} and \ref{application:graph-cut-single} (stated below). Further, one can construct instances in which the local search method can get stuck at subgraphs with contrasting densities. 

Indeed, there is a vast body of work in computer science that considers subgraph densities and graph cuts with an optimization objective, e.g., the densest subgraph problem \cite{lanciano2024survey}, the planted clique problem \cite{alon1998finding,feige2000finding}, minimum $k$-cut \cite{vazirani2001approximation}, and max $k$-cut \cite{frieze1997improved}. Our graph-theoretic results complement this extensive literature: Instead of addressing optimization problems over subgraph densities and cuts, we establish novel equitability bounds for these quantities. One can also juxtapose Szemer\'{e}di's regularity lemma \cite{szemeredi1978regular, frieze1999simple} with Theorem \ref{application:density-part}. Both the results obtain existential guarantees in the context of subgraph densities. However, the lemma addresses edge densities between the parts and requires them to be random-like. Hence, the result here is technically different from the regularity lemma. \\

\noindent {\bf Proximately-Dense Subgraphs.} We utilize one of our equitability results (see Remark \ref{remark:nonempty-condition} in Section \ref{section:eq-three-nonnegative}) to show that a graph can always be partitioned into $k$ subgraphs with approximately equal edge densities. In particular, for any given graph $G = (V,E)$ and vertex subset $S \subseteq V$, write density function $\rho(S)$ to denote half of the average degree within the subgraph induced by $S$, i.e., $\rho(S) \coloneqq \frac{1}{|S|} \  \big| \left\{ (u,v) \in E \mid u, v \in S \right\} \big|$. The density function $\rho: 2^V \mapsto \mathbb{R}_+$ is nonnegative.\footnote{However, the function $\rho$ is not subadditive.} We next observe that the marginal value of $\rho$ is at most $1$. Towards this, for any subset $S \subset V$, write $E(S)$ to denote the subset of edges both whose endpoints are in $S$, i.e., $E(S) \coloneqq \left\{ (u,v) \in E \mid u, v \in S \right\}$. Also, for any vertex $u \in V$, let $N(u) \subset V$ denote the neighbors of $u$. For any subset $S \subset V$ and vertex $u \in V \setminus S$, the marginal of $\rho$ is bounded as follows 
\begin{align*}
| \rho(S + u ) - \rho(S)| & = \left| \frac{|E(S + u )|}{|S + u|} - \frac{|E(S)|}{|S|} \right| \\ 
& = \left| \frac{|E(S)| \ + \ |N(u) \cap S|}{|S| + 1} - \frac{|E(S)|}{|S|} \right| \\
& = \left| \frac{|N(u) \cap S|}{|S| + 1}  -  \frac{|E(S)|}{|S| \left( |S| + 1 \right)}  \right| \\ 
& \leq \max \left\{ \frac{|N(u) \cap S|}{|S| + 1}, \ \frac{|E(S)|}{|S| \left( |S| + 1 \right)} \right\} \\  
& \leq \max \left\{ \frac{|S|}{|S| + 1}, \ \frac{|S|}{2 \left( |S| + 1 \right)} \right\} \tag{$|N(u) \cap S| \leq |S|$ and $|E(S)| \leq |S|^2/2$} \\   
& \leq \max \left\{ 1, \ \frac{1}{2} \right\} \\
& = 1.
\end{align*}
Hence, for any given integer $k \in \mathbb{Z}_+$, applying our equitability result for nonnegative valuations (see Theorem \ref{theorem:EQ3ApxComputation} and Remark \ref{remark:nonempty-condition}), with $n = k$, $m = |V|$, $v_i = \rho$, and $\Lambda = 1$, gives us the following theorem.   
\begin{theorem}\label{application:density-part}
Given any graph $G=(V,E)$ and any integer $k \leq |V|$, there always exists a partition $V_1, V_2, \ldots, V_k \neq \emptyset$ of the vertex set such that $| \rho(V_i) - \rho(V_j)| \leq 4$, for every $i,j \in [k]$. Further, such a $k$-partition can be computed in polynomial time.
\end{theorem} 

\noindent {\bf Equitable Graph Cuts.} Recall that for any (undirected) graph $G=(V,E)$, the cut function $\delta: 2^V \mapsto \mathbb{Z}_+$, counts the number of edges crossing vertex subsets $S \subseteq V$
$$\delta(S) \coloneqq \big| \left\{ (u,v) \in E \mid u \in S \text{ and } v \in V \setminus S \right\} \big|.$$ The cut function is nonnegative (and submodular). We can, hence, invoke out equitability result for indivisible items (Theorem \ref{theorem:EQ3ApxComputation}) and instantiate the nonnegative valuations $v_i = \delta$. Note that the maximum marginal value of the cut function is at most, $\Delta$, the maximum degree of the given graph $G$. Hence, for the valuations $v_i = \delta$, we have $\Lambda  = \Delta$. For any given integer $k \in \mathbb{Z}_+$, we can set the number of agents $n = k$ along with the number of items $m = |V|$ (i.e., aim to identify a $k$-partition of the vertex set) and apply the approximate equitability guarantee (Theorem \ref{theorem:EQ3ApxComputation}) to obtain the following result. 

\begin{theorem}\label{application:graph-cut-single}
Given any graph $G=(V,E)$ and integer $k \leq |V|$, there always exists a partition $V_1, V_2, \ldots, V_k \neq \emptyset$ of the vertex set such that $| \delta(V_i) - \delta(V_j)| \leq 5 \Delta +1$, for every $i,j \in [k]$; here, $\Delta$ is the maximum degree of $G$. Further, such a $k$-partition can be computed in polynomial time.
\end{theorem}

In the additive approximation factor obtained above, a dependence on the maximum degree, $\Delta$, is unavoidable. Consider, for example, a star graph and $k$ equal to the number of vertices. \\

\vspace*{5pt} 

\noindent {\bf Equitability under Quasilinear Utilities.} Here, we highlight the applicability our result for identical $\sigma$-subadditive valuations (Theorem \ref{theorem:EQ3ApxComputation-subadditive}). Consider a set of indivisible resources $N$ with a subadditive (or submodular) reward function $r : 2^N \mapsto \mathbb{R}_+$. In particular, $r(S) \geq 0 $ denotes the reward obtained via subset $S \subseteq N$. We also have a linear cost function $c: N \mapsto \mathbb{R}_+$ that specifies the cost of each resource in $N$. In this setup, for each subset $S \subseteq N$, the induced quasilinear utility, $u(S)$, is defined as $u(S) \coloneqq r(S) - c(S) = r(S) - \sum_{s \in S} c(s)$. Note that the utility $u(S)$ can be negative for certain subsets $S$ and, hence, might not be subadditive in the standard sense. The function $u$, however, is $\sigma$-subadditive; in particular, for any two disjoint subsets $S, T \subseteq N$ we have $u(S \cup T) = r(S \cup T) - c(S \cup T) = r(S \cup T) - c(S) - c(T) \leq r(S) + r(T) - c(S) - c(T) = u(S) + u(T)$. Hence, we obtain the following result as a corollary of our existential and algorithmic result for approximately equitable allocations under identical $\sigma$-subadditive valuations (\Cref{theorem:EQ3ApxComputation-subadditive}). 

\begin{theorem}\label{application:facility-partitioning}
Let $N$ be a set of indivisible resources and $u: 2^N \mapsto \mathbb{R}$ be an associated utility function that is normalized, $\sigma$-subadditive, and globally nonnegative ($u(N) \geq 0$). Then, for any given integer $k \leq |N|$, there always exists a partition $N_1, \ldots, N_k \neq \emptyset$ of the resources with approximately equal utilities: $\left| u(N_i) - u(N_j) \right| \leq 5 \Lambda + 1$, for all $i, j \in [k]$; here $\Lambda$ is the maximum marginal value of $u$. Further, such a $k$-partition can be computed in time polynomial in $|N|$ and $\Lambda$. 
 \end{theorem}

\section{Notation and Preliminaries}\label{section:prelims}

\paragraph{Indivisible Item Allocation.} We consider fair division of $m \in \mathbb{Z}_+$ indivisible items among a set of $n \in \mathbb{Z}_+$ agents. The set of items and agents will be denote by $[m] \coloneqq \{1, \ldots, m\}$ and $[n] \coloneqq \{1, \ldots, n\}$, respectively. An allocation $\alloc = (A_1, A_2, \ldots, A_n)$ of items among the $n$ agents is a {partition} of $[m]$ into $n$ pairwise disjoint subsets $A_i \subseteq [m]$. Here, $A_i$ represents the subset of items assigned to agent $i \in [n]$ and will be referred to as agent $i$'s bundle. Since $\alloc = (A_1, \ldots, A_n)$ is a partition of $[m]$, it satisfies $\cup_{i=1}^n A_i = [m]$ and $A_i \cap A_j = \emptyset$, for each $i \neq j$.

For any subset $B \subseteq [m]$ and item $\ell \in [m]$, we will use the shorthand $B + \ell$ to denote the set $B \cup \{ \ell \}$. Also, write $B + \ell + r$ to denote $B \cup \{\ell, r\}$. 

\paragraph{Agents' Valuations.} The cardinal preferences of the agents $i \in [n]$, over the indivisible items, are expressed via valuation functions $v_i : 2^{[m]} \rightarrow \mathbb{R}$. Specifically, $v_i(S)\in \mathbb{R}$ denotes the value that agent $i$ has for any subset of items $S \subseteq [m]$. Note that the agents' valuations are set functions and, hence, can require exponential (in $m$) space to be explicitly specified. To avoid such representation issues, and as is standard in the literature, we will assume that we have {value oracle} access to the $v_i$s, i.e., for each agent $i \in [n]$, we have an oracle that returns $v_i(S) \in \mathbb{R}$ in unit time, when queried with any subset $S \subseteq [m]$.

A discrete fair division instance will be denoted by the tuple $\langle [n], [m], \{ v_i \}_{i=1}^n \rangle$. We will assume, throughout, that the valuations of all the agents $i \in [n]$ are {normalized}, i.e., $v_i(\emptyset) = 0$, and globally nonnegative, $v_i([m]) \geq 0$. 

Notably, our results address fair division instances in which the agents' valuations are not necessarily monotone. Recall that a valuation $v$ is said to be monotone increasing if, for each item $j \in [m]$ and subset $S \subseteq [m]$, including $j$ in $S$ does not decrease the value, $v(S \cup \{ j \}) \geq v(S)$. Analogously, set function $c$ is said to be monotone decreasing, if $c(T \cup \{ x \}) \leq c(T)$, for each item $x \in [m]$ and subset $T \subseteq [m]$.

Going beyond monotone valuations, this work addresses the following two valuation classes:

\noindent
(i) {\bf Nonnegative valuations:} $v_i(S) \geq 0$ for each subset $S \subseteq [m]$ and all agents $i \in [n]$. 

\noindent 
(ii) {\bf $\sigma$-Subadditive valuations:}  $v_i(S \cup T) \leq v_i(S) + v_i(T)$ for all pairs of {\it disjoint} subsets $S, T \subseteq [m]$ and each agent $i \in [n]$.

Indeed, normalized nonnegative (and normalized $\sigma$-subadditive) set functions can be non-monotone.

\paragraph{Multilinear Extensions.} For each valuation $v_i : 2^{[m]} \rightarrow \mathbb{Q}$, we will write $V_i : [0,1]^m \rightarrow \mathbb{R}$ to denote its multilinear extension; in particular, for each $x \in [0,1]^m$, we have $V_i(x) \coloneqq \sum_{S \subseteq [m]} v_i(S)   \prod_{j \in S} x_j \allowbreak  \ \prod_{j \in [m] \setminus S} (1 - x_j)$. 
That is, $V_i(x)$ is the expected value, under valuation $v_i$, when the sets $S$ are drawn from the product distribution induced by $x$. Note that the multilinear extensions $\{V_i\}_{i=1}^n$ are continuous functions, since their definition only involves additions and multiplications of terms involving the components of $x$.

\paragraph{Cake Division.} Cake represents a divisible heterogeneous resource to be fairly divided among a set of $n$ agents. It is modeled as the interval $[0,1]$. A cake division $\mathcal{I} = (I_1, \ldots, I_n)$ is a partition of the cake $[0,1]$ into $n$ pairwise disjoint intervals $I_1,\ldots, I_n$. By convention, we will say that two intervals are disjoint if they intersect at finitely many points. Hence, intervals $[x,y]$ and $[y,z]$---with $0 \leq x \leq y \leq z \leq 1$---are said to be disjoint. Under division $\mathcal{I} = (I_1, \ldots, I_n)$, interval $I_i$ is assigned to agent $i$ and we have $\cup_{i=1}^n I_i = [0,1]$ along with $I_i \cap I_j = \emptyset$ for all $i \neq j$. The cake-division literature, in general, also considers partitions in which each agent can be assigned a finite union of intervals. The current work, however, focusses exclusively on {contiguous} cake division wherein each agent $i$ receives a single connected interval $I_i$. 

In the context of cake division, the preferences of the agents are specified via valuations $f_i$ defined over intervals $I = [x,y]$, with $0 \leq x \leq y\leq 1$. In particular, $f_i(I) \in \mathbb{R}$ denotes that value that agent $i \in [n]$ has for the cake interval $I \subseteq [0,1]$. A cake division instance will be denoted by a tuple $\langle [n], \{f_i\}_{i=1}^n \rangle$. The paper, throughout, conforms to cake valuations $f_i$ that are normalized, $f_i(\emptyset) = 0$, globally nonnegative, $f_i([0,1]) \geq 0$, and continuous. Also, for any interval $I = [x,y]$, we will write $\len(I)$ to denote the length of the interval, $\len(I) = y -x$. 

The current work does not assume that the valuations $\{f_i\}_{i=1}^n$ are monotone, i.e., we do not require $f_i(I) \leq f_i(J)$, for intervals $I \subseteq J$. Specifically, the paper addresses cake division under valuations $\{ f_i\}_{i\in [n]}$ that are either 

\noindent 
(i) {\bf Nonnegative:} $f_i(I) \geq 0$ for all intervals $I \subseteq [0,1]$.

\noindent 
(ii) {\bf $\sigma$-Subadditive:} $f_i( I \cup J) \leq f_i(I) + f_i(J)$ for all disjoint intervals $I, J \subseteq [0,1]$. 

Recall that a valuation $g$ is said to be $\sigma$-additive if $g(I \cup J) = g(I) + g(J)$ for all disjoint intervals $I, J \subseteq [0,1]$. Hence, $\sigma$-subadditivity generalizes $\sigma$-additivity. Also, a $\sigma$-subadditive function can be non-monotone and even negative. In fact, our cake division results hold for a more general class, which we refer to as $c\sigma$-subadditive: a valuation $f$ is said to be {\it $c\sigma$-subadditive} if for any two {\it contiguous} intervals $I, J \subseteq [0,1]$ (i.e, for any $I = [x,y]$ and $J = [y,z]$, with $0 \leq x \leq y \leq z \leq 1$) we have $f(I \cup J) \leq f(I) + f(J)$. Indeed, every $\sigma$-subadditive function is  $c\sigma$-subadditive. 

\paragraph{Envy-Freeness.} In the cake-division context, in any instance, $\langle [n], \{f_i\}_{i=1}^n \rangle$, a cake division $\mathcal{I} = (I_1, \ldots, I_n)$ is said to be \emph{envy-free} ($\EF$) if, for every pair of agents $i,j \in [n]$, we have $f_i(I_i) \geq f_i(I_j)$. 

In the indivisible-items setting and for any integer $k \geq 1$, an allocation $\alloc = (A_1, \ldots, A_n)$ is said to be \emph{envy-free up to $k$ edits} ($\EFk$) if the envy that any agent $i\in [n]$ has towards any other agent $j \in [n]$ can be resolved by shifting at most $k$ items in or out of $A_i$ and $A_j$ combined. Formally,\footnote{Recall that $X \triangle Y$ denotes the symmetric difference between the two sets, $X \triangle Y \coloneqq \left(X \setminus Y \right) \cup \left(Y \setminus X \right)$.} 
\begin{definition}[$\EFk$]\label{definition:EFk}
In a fair division instance $\langle [n], [m], \{ v_i \}_{i=1}^n \rangle$, an allocation $\alloc = (A_1, \ldots, A_n)$ is said to be envy-free up to $k$ edits if, for each pair of agents $i, j$, there exists (notional and nearby) subsets $A'_i, A'_j \subseteq [m]$ with the properties that $v_i(A'_i) \geq v_i(A'_j)$ and $| A_i \triangle A'_i| + |A_j \triangle A'_j| \leq k$. 
\end{definition}
We will establish existence of $\EFt$ allocations (i.e., we will address the $k=3$ case) under valuations that are not necessarily monotone. 

Note that the relaxations of envy-freeness formulated in discrete fair division---specifically, $\EFo$---capture the idea that a bounded number of changes in the assigned bundles suffice to achieve envy-freeness. $\EFk$ conforms to this paradigm. Indeed, an allocation (of mixed manna) is said to be $\EFo$ if any existing envy can be mitigated by the deletion of one item from $A_i$ and $A_j$ combined; see, e.g., \cite{aziz2022fair}. Definition \ref{definition:EFk} puts forward a further relaxation by allowing for $k$ edits, i.e., $k$ additions and deletions from $A_i$ and $A_j$ combined.  Also, observe that the following implication: $\EFo$ implies $\mathrm{EFE}1$. 

\paragraph{Equitability.}
A cake division $(I_1, \ldots, I_n)$ is said to be \emph{equitable} ($\mathrm{EQ}$) if the equality $f_i(I_i) = f_j(I_j)$ holds for each pair of agents $i,j \in [n]$.

\begin{definition}[$\EQk$]\label{definition:EQk}
In a fair division instance $\langle [n], [m], \{ v_i \}_{i=1}^n \rangle$, an allocation $\alloc = (A_1, \ldots, A_n)$ is said to be equitable up to $k$ edits if, for each pair of agents $i, j$, there exists subsets $A'_i, A'_j \subseteq [m]$ with the properties that $v_i(A'_i) \geq v_j(A'_j)$ and $| A_i \triangle A'_i| + |A_j \triangle A'_j| \leq k$. 
\end{definition}

\section{Tools and Techniques}
We will utilize two well-known results from Combinatorial Topology, the Brouwer's fixed point theorem and Sperner's lemma; both these results are stated next. 

\begin{theorem}[Brouwer's Fixed-Point Theorem \cite{border1985fixed}]\label{theorem:brouwer-fpt} Every continuous function $f: K \mapsto K$ from a nonempty, convex, compact set $K \subset \mathbb{R}^d$ to $K$ itself has a fixed point, i.e., there exists an $x^* \in K$ such that $f(x^*) = x^*$. 
\end{theorem}

An \emph{$\ell$-simplex} $S$ is convex hull of $(\ell+1)$ affinely independent vectors, $u_1, u_2, \ldots, u_{\ell+1} \in \mathbb{R}^d$, with $d \geq \ell$. A \emph{$k$-face} of an $\ell$-simplex $S$, for any $k \leq \ell$, is a $k$-simplex formed by any $(k+1)$ vertices of $S$. 

A \emph{triangulation} $\mathcal{T}$ of an $\ell$-simplex $S$ is a collection of smaller $\ell$-simplices whose union is $S$ and any two of them are either disjoint or intersect at a face common to both. The smaller simplices of a triangulation are called the \emph{elementary simplices} of the triangulation. In addition, the collection of all the vertices of the elementary simplices is called the \emph{vertices of the triangulation}; this set of vertices is denoted as $\mathrm{vert}(\mathcal{T})$.

For an $\ell$-simplex $S = \mathrm{conv}\left( \left\{u_1, \ldots, u_{\ell+1} \right\} \right)$ all the $(\ell-1)$-faces of $S$ are called its \emph{facets}. Note that simplex $S$ has $(\ell+1)$ facets. Moreover, we will follow the indexing wherein, for each index $1 \leq j \leq (\ell+1)$, the $j$th facet is the convex hull of all the vertices, except $u_j$. That is, the $j$th facet of $S$ is $\mathrm{conv}\left( \left\{u_1, \ldots, u_{j-1}, u_{j+1}, \ldots, u_{\ell+1} \right\} \right)$.  

A \emph{labeling} $\lambda$ of a triangulation of a $\ell$-simplex is a mapping from the vertices of the triangulation to numbers $\{1, 2, \ldots, \ell+1\}$, i.e., $\lambda: \mathrm{vert}(\mathcal{T}) \mapsto [\ell +1]$.  A {labeling} is called  a \emph{Sperner's labeling} if it obeys the following rule: each vertex $w \in \mathrm{vert}(\mathcal{T})$ of the triangulation that lies on the $j$th facet is not assigned the label $j$. Equivalently, if vertex $w \in  \mathrm{conv}\left( \left\{u_1, \ldots, u_{j-1}, u_{j+1}, \ldots, u_{\ell+1} \right\} \right)$, then $\lambda(w) \neq j$. Any vertex of the triangulation that lies in the strict interior of $S$ (i.e., not on any of its facets) can have any label. 
 
 Finally, an elementary simplex of a triangulation is said to be \emph{fully labeled} if all of its vertices have distinct labels. 

\begin{theorem}[Sperner's Lemma]\label{theorem:sperners-lemma} Let $\lambda:  \mathrm{vert}(\mathcal{T}) \mapsto [\ell+1]$ be a Sperner labeling of a triangulation $\mathcal{T}$ of an $\ell$-simplex. Then, there exists a fully-labeled elementary simplex in $\mathcal{T}$. 
\end{theorem} 

We will also utilize the result of Su~\cite{Su1999rental} that establishes the existence of envy-free cake divisions under relevant conditions on the agents' valuations $f_1,\ldots, f_n$. In particular, a key prerequisite for the result---referred to as the hungry condition---is that the agents value any nonempty interval more than an empty piece of the cake. Specifically, for any agent $i \in [n]$, for any $I \subseteq [0,1]$ with $I \neq \emptyset$ we have $f_i(I) > 0 = f_i(\emptyset)$; the last equality follows from $f_i$ being normalized. The existence of fair cake divisions holds under even weaker versions of the condition; see, e.g., \cite{meunier2019envy} and \cite{bhaskar2025connected}. Though, the current work invokes the hungry condition as stated above. 
 
\begin{theorem}[Envy-Free Cake Division~\cite{Su1999rental}] \label{theorem:cake-sperner-ef}
Every cake-division instance $\langle [n], \{f_i\}_{i=1}^n \rangle$ in which the agents' valuations, $f_1, \ldots, f_n$, are continuous and satisfy the hungry condition, admits a contiguous envy-free division $\mathcal{I}=(I_1,\ldots, I_n)$.  
\end{theorem}

\subsection{Cake Construction}\label{section:cake-reduction}
A key technique that we will use throughout the paper is to convert a discrete fair division instance (i.e., an instance with indivisible items) to a cake division instance; we will refer to this step as cake construction. Along with such a construction step, we also develop a rounding method (\Cref{section:CakeRounding}) that starts with a fair cake division $\calI=(I_1, \ldots, I_n)$ (in a constructed cake instance) and yields a fair allocation $\calA=(A_1, \ldots, A_n)$ in the original instance with indivisible items. We will refer to this method as cake rounding. 

As mentioned, the cake construction process (\Cref{algo:cake-reduction}) takes as input a discrete fair division instance $\langle [n], [m], \{ v_i \}_{i=1}^n \rangle$ and returns a cake division instance $\langle [n], \{ f_i \}_{i=1}^n \rangle$. To form the cake, the items $1,2 \ldots, m$ are `placed' contiguously on the interval $[0,1]$; in particular, each item $k \in [m]$ is associated with the the sub-interval $[\frac{k-1}{m},\frac{k}{m}]$. 

To define the cake valuations, $f_i$, we first define a linear function $b$ that maps intervals $I = [x,y] \subseteq [0,1]$ to vectors in $[0,1]^m$. Specifically, $b(I) \coloneqq (b_1(I), \ldots, b_m(I)) \in [0,1]^m$, where each component $b_k(I) \in [0,1]$ denotes the fraction of item $k \in [m]$ that is contained in the interval $I$. That is,  $b_k(I) = m \  \len \left(I \cap \left[\frac{k-1}{m}, \frac{k}{m} \right] \right)$, for each $k \in [m]$.\footnote{Recall that $\len(J)$ denotes the length of the interval $J$.} Note that the function $b$ is additive: $b(I_1 \cup I_2) = b(I_1) + b(I_2)$ for any two disjoint intervals $I_1, I_2 \subseteq [0,1]$.

Given $b$, the cake valuation $f_i$ of each agent $i \in [n]$ is defined as the composition of $V_i$ (the multilinear extension of $v_i$) with $b$, i.e., $f_i \coloneqq V_i \circ b$. Note that $f_i$ takes as an argument an interval $I = [x,y] \subseteq [0,1]$ and maps it to a real number. Further, each $f_i$ is continuous,  since both $V_i$ and $b$ are continuous.  
\floatname{algorithm}{Algorithm}
\begin{algorithm}[H]
\caption{\textsc{CakeConstruction}} \label{algo:cake-reduction}
\begin{tabularx}{\textwidth}{X}
{\bf Input:} A discrete fair division instance $\langle [n], [m], \{ v_i \}_{i=1}^n \rangle$. 

{\bf Output:} A cake division instance $\langle [n], \{ f_i \}_{i=1}^n \rangle$.
\end{tabularx}
  \begin{algorithmic}[1]
  		\STATE Associate with each item $k \in [m]$ the sub-interval $[\frac{k-1}{m},\frac{k}{m}]$.
		\STATE Define function $b$ over intervals $I = [x,y] \subseteq [0,1]$, such that $b(I) = (b_1(I), \ldots, b_m(I)) \in [0,1]^m$ and each $b_k(I) = m \ \len(I \cap [\frac{k-1}{m},\frac{k}{m}])$. 
		\STATE For each agent $i \in [n]$, set the cake valuation $f_i \coloneqq V_i \circ b$, where $V_i$ is the multilinear extension of $v_i$.
		\RETURN $\langle [n], \{ f_i \}_{i=1}^n \rangle$.
		\end{algorithmic}
\end{algorithm}

\subsection{Cake Rounding}\label{section:CakeRounding}
\Cref{algo:cake-rounding} (detailed next) and the subsequent two lemmas (\Cref{lemma:cake-rounding} and \Cref{lemma:cake-rounding-2}) provide the rounding method and its guarantees. \Cref{lemma:cake-rounding} addresses envy-freeness, while \Cref{lemma:cake-rounding-2} covers equitability. 

The \textsc{CakeRounding} algorithm takes as input a discrete fair division instance $\langle [n], [m], \{v_i\}_{i=1}^n \rangle$, the constructed cake division instance $\langle [n], \{f_i\}_{i=1}^n\rangle$ (obtained via \textsc{CakeConstruction}), and in it a contiguous cake division $\mathcal{I} = (I_1, \ldots, I_n)$. The algorithm rounds $\mathcal{I}$ to an allocation $\alloc = (A_1, \ldots, A_n)$ while preserving approximate fairness (envy-free and equitability). 

For ease of exposition, we reindex the agents such that, in $\calI=(I_1, \ldots, I_n)$, the $i$th interval from the left in the cake is the one assigned to agent $i$. That is, agent $1$'s interval, $I_1$, is the left-most interval among $\{I_1, \ldots, I_n\}$, and $I_n$ is the right-most. Note that such a reindexing holds without loss of generality and does not impact the fairness guarantees. 
 
For each agent $i \in [n]$, we define $B_i \coloneqq \{ k \in [m] : b_k(I_i) = 1\}$ as the subset of items that are completely contained in $I_i$. Further, define $\ell_i$ and $r_i$ to be the \emph{fractionally covered} items at the left and right ends of the interval $I_i$, i.e., item $\ell_i \coloneqq \min \{k \in [m] \mid  0 < b_k(I_i) < 1\}$ and item $r_i \coloneqq  \max \{k  \in [m] \mid 0 < b_k(I_i) < 1\}$. Note that it could be the case that $\ell_i = r_i$. Also, if the left (right) endpoint of interval $I_i$ falls at $\frac{k-1}{m}$, for any integer $k \in [m]$, then $\ell_i$ ($r_i$) does not exist.   

\Cref{algo:cake-rounding} starts by initializing all the bundles $A_i = \emptyset$ and the set of unassigned items $U = [m]$. Further, the algorithm, in its for-loop (Lines \ref{line:cake-rouding-loop-begin}-\ref{line:cake-rouding-loop-end}), rounds the given cake division $\mathcal{I}$ from left to right: first the interval $I_1$ is rounded to bundle $A_1 \subseteq [m]$, then $I_2$ is rounded to $A_2 \subseteq [m]$, and so on. Throughout, the algorithm maintains the following key property: for each agent $i \in [n]$, the assigned bundle $A_i$ satisfies $|A_i \triangle A^*_i| \leq 1$, where $A^*_i$ is the maximum-valued set among the following four: $B_i$, $B_i + \ell_i$, $B_i + r_i$, and $B_i + \ell_i + r_i$.\footnote{That is, $A^*_i \in \argmax\{v_i(X) : X \in \{B_i, B_i \cup \{\ell_i\}, B_i \cup \{r_i\}, B_i \cup \{\ell_i, r_i\}\}\}$. Also, recall that $B + \ell$ denotes $B \cup \{ \ell \}$.} This property is achieved by construction -- see the assignment table in Line \ref{line:table}. Using this property, in the proof of the following lemma we will show that the assigned subsets $A_i$ constitute an allocation (i.e., these bundles are pairwise disjoint and partition $[m]$) and inherit the desired fairness properties from the intervals $I_i$. 

\floatname{algorithm}{Algorithm}
\begin{algorithm}[H]
\caption{\textsc{CakeRounding}} \label{algo:cake-rounding}
\begin{tabularx}{\textwidth}{X}
{\bf Input:} A discrete fair division instance $\langle [n], [m], \{v_i\}_{i=1}^n \rangle$ and a (contiguous) cake division $\mathcal{I} = (I_1, \ldots, I_n)$ in the constructed cake division instance $\langle [n], \{f_i\}_{i=1}^n\rangle$.\\
{\bf Output:} Allocation $\alloc = (A_1, \ldots, A_n)$.
\end{tabularx}
  \begin{algorithmic}[1]
  		\STATE Reindex the agents such that, for each $i \in [n]$, interval $I_i$ is the $i$th interval from the left, among $I_1,\ldots, I_n$. 
		\STATE For each agent $i \in [n]$, set $B_i \coloneqq \{ k \in [m] : b_k(I_i) = 1\}$ along with items $\ell_i \coloneqq \min \{k \in [m] \mid  0 < b_k(I_i) < 1\}$ and $r_i \coloneqq   \max \{k  \in [m] \mid 0 < b_k(I_i) < 1\}$. \label{line:define-B-i}
		\STATE Also, for each agent $i \in [n]$, set $A^*_i \in \argmax \Big\{v_i(X)  \mid X \in \{ B_i, B_i + \ell_i, B_i + r_i, B_i + \ell_i + r_i \} \Big\}$. \label{line:Astar}
		\STATE Initialize (unassigned items) $U = [m]$ and bundles $A_i = \emptyset$ for all $i \in [n]$. \label{line:rounding-init}
		\FOR{agents $i = 1$ to $n$} \label{line:cake-rouding-loop-begin}
		\STATE Set $A_i$ based on the eight exhaustive cases listed in the following table. That is, set $A_i$ based on whether $\ell_i \in U$ or $\ell_i \notin U$, and which among the four choices in Line \ref{line:Astar} yields $A^*_i$.  \label{line:table}
			\vspace*{7pt}
		\begin{center}
		\begin{tabular}{ |c|c|c|c|c| } 
		\hline
 	 	\ 		  & $A^*_i = B_i$ & $A^*_i = B_i + \ell_i $ & $A^*_i = B_i + r_i$ & $A^*_i = B_i + \ell_i + r_i$ \\ 
		\hline
 		If $\ell_i \in U$ & $A_i = B_i + \ell_i$ &  $A_i = B_i + \ell_i$ & $A_i = B_i + \ell_i + r_i$ & $A_i = B_i + \ell_i + r_i$ \\ 
 		\hline  
		If $\ell_i \notin U$ & $A_i = B_i$ &  $A_i = B_i $ & $A_i = B_i + r_i$ & $A_i = B_i + r_i$ \\
		\hline
		\end{tabular}
		\end{center}  			
			\vspace*{7pt}
	\STATE Update $U \gets U \setminus A_i$.  
    	\ENDFOR \label{line:cake-rouding-loop-end}
    	
		\RETURN $\alloc = \left( A_1, \ldots, A_n \right)$
		\end{algorithmic}
\end{algorithm}

\begin{lemma}[$\EF$ Rounding Lemma]\label{lemma:cake-rounding}
Let $\langle [n], [m], \{v_i\}_{i=1}^n \rangle$ be a discrete fair division instance and $\mathcal{I} = (I_1, \ldots, I_n)$ be a contiguous cake division in the cake division instance $\langle [n], \{f_i\}_{i=1}^n\rangle$ obtained via \textsc{CakeConstruction} (\Cref{algo:cake-reduction}). Also, for parameter $\alpha \geq 0$, let $\calI$ satisfy $f_i(I_i) + \alpha \geq f_i(I_j)$ for all $i,j \in [n]$. Then, given these instances and the  division $\calI$ as input, \Cref{algo:cake-rounding} computes an allocation $\calA = (A_1, \ldots, A_n)$ wherein, for each pair of agents $i, j \in [n]$, there exists subsets $A'_i, A'_j \subseteq [m]$ such that $v_i(A'_i) + \alpha \geq v_i(A'_j)$ and $| A_i \triangle A'_i| + |A_j \triangle A'_j| \leq 3$. 
\end{lemma}
\begin{proof}
\Cref{algo:cake-rounding}, in each iteration $i \in [n]$ of its for-loop, always includes the set $B_i \coloneqq \{k \in [m] \mid b_k(I_i) = 1\}$ into agent $i$'s bundle $A_i$. That is, we always have $B_i \subseteq A_i$. The algorithm follows the assignment table of Line \ref{line:table} and decides on the inclusion of items $\ell_i \coloneqq \min \{k \in [m] \mid  0 < b_k(I_i) < 1\}$ and $r_i \coloneqq  \max \{k  \in [m] \mid 0 < b_k(I_i) < 1\}$. In particular, if $\ell_i \in U$ when $A_i$ is being populated, then the algorithm necessarily includes this item in the bundle; see the first row of the assignment table in Line \ref{line:table}. This ensures that all the items are assigned. Further, if $\ell_i \notin U$ when $A_i$ is being populated, then the algorithm leaves this item out of $A_i$; see the second row of the table.  These design choices ensure that the bundles $A_i$ indeed constitute an allocation, i.e., these bundles  partition $[m]$ and are pairwise disjoint. 

Moreover, the design of the assignment table ensures that, for each agent $i \in [n]$, and irrespective of the case in Line \ref{line:table}, we have $|A_i \triangle A^*_i| \leq 1$. 
 
 Next, recall that the cake valuation $f_i$ is defined as $f_i = V_i \circ b$. Hence, by the lemma assumption, for every pair of agents $i, j \in [n]$, it holds that \begin{align}
V_i(b(I_i)) + \alpha \geq V_i(b(I_j)) \label{ineq:Vf}
\end{align} 
Additionally, by the definition of the set $B_i$, we have $b_k(I_i) = 1$ for all items $k \in B_i$. Also, $b_k(I_i) \in (0,1)$ for $k \in \{\ell_i, r_i\}$. Similarly, $b_k(I_j) = 1$ for all items $k \in B_j$, and $b_k(I_j) \in (0,1)$ for $k \in \{\ell_j, r_j\}$. Furthermore, given that each $I_i$ is an interval, for any $k' \notin B_i \cup\{\ell_i, r_i\}$, the sub-interval $[\frac{k'-1}{m},\frac{k'}{m}]$ does not intersect with $I_i$, i.e., $b_{k'}(I_i) = 0$. These observations and the fact that $V_i$ is the multilinear extension of $v_i$ implies that $V_i(b(I_i))$ is a convex combination of the following four values: $v_i(B_i)$, $v_i(B_i + \ell_i)$, $v_i(B_i + r_i)$, and $v_i(B_i + \ell_i + r_i)$. Analogously, $V_i(b(I_j))$ is a convex combination of $v_i(B_j)$, $v_i(B_j + \ell_j)$, $v_i(B_j + r_j)$, and $v_i(B_j + \ell_j + r_j)$. Therefore, 
\begin{align}\label{equation:MLE-condition}
v_i(A^*_i) + \alpha & = \max \Big\{v_i(B_i), v_i(B_i + \ell_i), v_i(B_i + r_i), v_i(B_i + \ell_i + r_i) \Big\} + \alpha \nonumber \tag{by defn.~of $A^*_i$}\\ 
 & \geq V_i(b(I_i)) + \alpha  \nonumber \\ 
 & \geq V_i(b(I_j)) \nonumber \tag{via (\ref{ineq:Vf})} \\ 
 & \geq \min \Big\{v_i(B_j), v_i(B_j + \ell_j), v_i(B_j + r_j), v_i(B_j + \ell_j + r_j) \Big\} \label{ineq:ViVj}
\end{align}

Consider any $A'_j \in \argmin\Big\{ v_i(X) \mid X \in \{ B_j, B_j + \ell_j, B_j + r_j, B_j + \ell_j + r_j\} \Big\}$. Hence, the minimum on the right-hand-side of inequality (\ref{ineq:ViVj}) is attained at $v_i(A'_j)$. Note that, by construction, the bundle, $A_j$, assigned to agent $j$ satisfies $B_j \subseteq A_j \subseteq B_j + \ell_j + r_j$. Hence, $|A_j \triangle A'_j| \leq 2$.  

Finally, with this set $A'_j$ in hand and setting $A'_i = A^*_i$, we obtain the two desired properties: 

\noindent 
(i) $v_i(A'_i) + \alpha \geq v_i(A'_j)$; see inequality (\ref{ineq:ViVj}), and  

\noindent 
(ii) $|A'_i \triangle A_i| + |A'_j \triangle A_j| \leq 1 + 2 = 3$. 

The lemma stands proved. 
\end{proof}

The proof of the following lemma is similar to the one above and, hence, is omitted. 

\begin{lemma}[$\EQ$ Rounding Lemma]\label{lemma:cake-rounding-2}
Let $\langle [n], [m], \{v_i\}_{i=1}^n \rangle$ be a discrete fair division instance and $\mathcal{I} = (I_1, \ldots, I_n)$ be a contiguous cake division in the cake division instance $\langle [n], \{f_i\}_{i=1}^n\rangle$ obtained via \textsc{CakeConstruction} (\Cref{algo:cake-reduction}). Also, for parameter $\alpha \geq 0$, let $\calI$ satisfy $f_i(I_i) + \alpha \geq f_j(I_j)$ for all $i,j \in [n]$. Then, given these instances and the  division $\calI$ as input, \Cref{algo:cake-rounding} computes an allocation $\calA = (A_1, \ldots, A_n)$ wherein, for each pair of agents $i, j \in [n]$, there exists subsets $A'_i, A'_j \subseteq [m]$ such that $v_i(A'_i) + \alpha \geq v_j(A'_j)$ and $| A_i \triangle A'_i| + |A_j \triangle A'_j| \leq 3$. 
\end{lemma}
\section{Envy-Freeness Under $\sigma$-Subadditive Valuations}
This section develops our envy-freeness guarantees for $\sigma$-subadditive valuations. The valuations addressed here can be negative and non-monotone. 

We begin by establishing in Section \ref{section:comb-multilin-subadd} that the multilinear extension of a $\sigma$-subadditive set function also has subadditive properties (\Cref{lemma:subadditive-mle}). Next, in Section \ref{section:ef-cake-subadd} we prove that contiguous envy-free cake divisions always exist under $\sigma$-subadditive and globally nonnegative ($f_i([0,1] \geq 0$) valuations. Note that, under such valuations, certain sub-intervals can have negative values, i.e., parts of the cake can be burnt. Hence, it is not a given that the hungry condition holds, and one can apply \Cref{theorem:cake-sperner-ef} to obtain the existence of envy-free cake divisions. In fact, envy-free divisions are known to {\it not} exist for burnt cakes in general~\cite{avvakumov2021envy}. We overcome these barriers and prove that envy-free cake divisions do exist when the cake valuations $f_i$ are $\sigma$-subadditive and globally nonnegative . 

In fact, our envy-freeness guarantee holds more generally, for valuations $f_i$ that are $c\sigma$-subadditive (see Section \ref{section:prelims} for a definition). We obtain the result via an application of Sperner's lemma (\Cref{theorem:cake-sperner-ef}). 

Section \ref{section:ef-three-subadditive} builds on the cake division result for the indivisible-items context. Here, we prove that $\EFt$ allocations exist under valuations $v_i$ that are $\sigma$-subadditive and globally nonnegative ($v_i([m] \geq 0$). The valuations $v_i$s can be negative for certain subsets.

\subsection{Combinatorial to Multilinear Subadditivity}
\label{section:comb-multilin-subadd}

The theorem below considers a $\sigma$-subadditive (possibly negative and non-monotone) set function $f: 2^{[m]} \rightarrow \mathbb{R}$ and its multilinear extension $F:[0,1]^m \rightarrow \mathbb{R}$. The result asserts that for any $a,b \in [0,1]^m$, with $a+b \in [0,1]^m$, we have $F(a + b) \leq F(a) + F(b)$. That is, under the stated conditions, $F$ also behaves subadditively. Recall that for a $\sigma$-subadditive set function $f$, by definition, we have $f(S \cup T) \leq f(S) + f(T)$ for all {\it disjoint} subsets $S,T$. This result is formally stated and proved via an inductive argument below.

\begin{theorem}\label{lemma:subadditive-mle}
 Let $f: 2^{[m]} \rightarrow \mathbb{R}$ be a $\sigma$-subadditive set function and $F:[0,1]^m \rightarrow \mathbb{R}$ be its multilinear extension. If $a,b,c \in [0,1]^m$ satisfy $c = a+b$, then $F(c) \leq F(a) + F(b)$. 
\end{theorem}
\begin{proof}
	Define $\mathrm{supp}(a) \coloneqq \{j \in [m] \mid a_j > 0\}$ and $\mathrm{supp}(b) \coloneqq \{j \in [m] \mid b_j > 0\}$ to be the set of nonnegative components of $a$ and $b$, respectively. We will establish the desired inequality $F(c) \leq F(a) + F(b)$ for all $a,b,c \in [0,1]^m$, with $c = a+b$, via a proof by induction. Specifically, we write $\mathcal{P}(k)$ to denote that proposition that the desired inequality holds for any $a, b, c \in [0,1]^m$ satisfying $c = a + b$ and the condition that $|\mathrm{supp}(a) \cap \mathrm{supp}(b)| = k$. Inducting on $k$, we will prove that the proposition $\mathcal{P}(k)$ holds for all $k \in \{0,1,\ldots, m\}$. 
	
We will first prove the base case, i.e., show that $\mathcal{P}(0)$ holds. Then, for induction, we will show that for any $1 \leq k \leq m$, if the proposition $\mathcal{P}(k-1)$ is true, then so is $\mathcal{P}(k)$. Together, these steps will complete the proof.
	
\paragraph{Base case:} $\mathcal{P}(k)$ for $k=0$. Since $k=0$, it holds that $\mathrm{supp}(a) \cap \mathrm{supp}(b) = \emptyset$. Hence, the set $[m]$ is partitioned by the following three pairwise disjoint subsets: $\mathrm{supp}(a) , \mathrm{supp}(b)$, and $R \coloneqq [m] \setminus \left( \mathrm{supp}(a) \cup \mathrm{supp}(b) \right)$. To prove $F(c) \leq F(a) + F(b)$ for this case, we begin by noting that, by definition: 
\begin{align}
F(c)  = \sum_{S\subseteq [m]} f(S)  \ \prod\limits_{j \in S} c_j  \  \prod\limits_{j \in [m] \setminus S} (1-c_j) \label{ineq:defnMult}
\end{align}

Since $[m]$ is partitioned by subsets $\mathrm{supp}(a)$, $\mathrm{supp}(b)$, and $R$, the outer summation over $S \subseteq [m]$ in equation (\ref{ineq:defnMult}) can be replaced with three nested summations over subsets $A \subseteq \mathrm{supp}(a)$, $B \subseteq \mathrm{supp}(b)$, and $C \subseteq R$. Hence, 
\begin{align}
F(c) & = \sum_{S\subseteq [m]} f(S) \ \prod\limits_{j \in S} c_j \ \prod\limits_{j \in [m] \setminus S} (1-c_j) \nonumber \\
& = \sum_{A\subseteq \mathrm{supp}(a)} \sum_{B\subseteq \mathrm{supp}(b)} \sum_{C\subseteq R} f(A \cup B \cup C) \  \prod\limits_{j \in A \cup B \cup C} c_j \  \prod\limits_{j \in [m] \setminus (A \cup B \cup C)} (1-c_j) \label{equation:subadditivity-lemma-multilinear-extension}
\end{align}
Using again the fact that $[m]$ is partitioned into $\mathrm{supp}(a)$, $\mathrm{supp}(b)$, and $R$, we get that the $[m] \setminus (A \cup B \cup C)$ is itself the union of three disjoint sets: $\mathrm{supp}(a) \setminus A$, $\mathrm{supp}(b) \setminus B$, and $R \setminus C$. Using this observation, we expand the term in the right-hand-side summand of equation (\ref{equation:subadditivity-lemma-multilinear-extension}) as follows
\begin{align}
& f(A \cup B \cup C) \  \prod\limits_{j \in A \cup B \cup C} c_j \  \prod\limits_{j \in [m] \setminus (A \cup B \cup C)} (1-c_j)  \nonumber \\
& \ \ = f(A \cup B \cup C) \  \prod\limits_{j \in A \cup B \cup C} c_j  \  \prod\limits_{j \in \mathrm{supp}(a) \setminus A} (1-c_j) \  \prod\limits_{j \in \mathrm{supp}(b) \setminus B} (1-c_j) \  \prod\limits_{j \in R \setminus C} (1-c_j) \label{ineq:interimTerm}
\end{align}
Further, $c_j = 0$ for each $j \in R$. This follows from the fact that $c=a+b$ and $R = [m] \setminus \left( \mathrm{supp}(a) \cup \mathrm{supp}(b) \right)$. Hence, all the terms where $C \neq \emptyset$ have a zero contribution in the right-hand-side summation of equation (\ref{equation:subadditivity-lemma-multilinear-extension}). Equivalently, the terms---given in equation (\ref{ineq:interimTerm})---that have a nonzero contribution in the summation must have $C = \emptyset$ and, hence, are of the form  
\begin{align}
& f(A \cup B)  \  \prod\limits_{j \in A \cup B} c_j \ \prod\limits_{j \in \mathrm{supp}(a) \setminus A} (1-c_j) \  \prod\limits_{j \in \mathrm{supp}(b) \setminus B} (1-c_j) \  \prod\limits_{j \in R} (1-c_j) \nonumber \\
& = f(A \cup B) \ \prod\limits_{j \in A \cup B} c_j \  \prod\limits_{j \in \mathrm{supp}(a) \setminus A} (1-c_j) \  \prod\limits_{j \in \mathrm{supp}(b) \setminus B} (1-c_j)  \label{equation:subadditivity-lemma-multilinear-extension-simplified}
\end{align}
Here, the last equality holds since $\prod\limits_{j \in R} (1-c_j) = 1$; recall that $c_j = 0$ for all $j \in R$. 

Next, note that sets $A$ and $B$ are always disjoint, since $A \subseteq \mathrm{supp}(a)$, $B \subseteq \mathrm{supp}(b)$, and $\mathrm{supp}(a) \cap \mathrm{supp}(b) = \emptyset$. Hence, the $\sigma$-subadditivity of $f$ gives us $f(A \cup B) \leq f(A) + f(B)$. Using this inequality along with equations (\ref{equation:subadditivity-lemma-multilinear-extension}) and (\ref{equation:subadditivity-lemma-multilinear-extension-simplified}), we obtain 
\begin{align}
F(c) & \leq \sum_{A\subseteq \mathrm{supp}(a)} \sum_{B\subseteq \mathrm{supp}(b)}  (f(A) +  f(B)) \  \prod\limits_{j \in A \cup B} c_j \  \prod\limits_{j \in \mathrm{supp}(a) \setminus A} (1-c_j) \  \prod\limits_{j \in \mathrm{supp}(b) \setminus B} (1-c_j) \nonumber \\
& = \sum_{A\subseteq \mathrm{supp}(a)}  f(A) \  \prod\limits_{j \in A} c_j \  \prod\limits_{j \in \mathrm{supp}(a) \setminus A} (1-c_j) \  \left( \sum_{B\subseteq \mathrm{supp}(b)} \prod\limits_{j \in B} c_j \  \prod\limits_{j \in \mathrm{supp}(b) \setminus B} (1-c_j) \right) \nonumber \\
& \ \ \ \ \ \ \ + \sum_{B\subseteq \mathrm{supp}(b)} f(B) \  \prod\limits_{j \in B} c_j \  \prod\limits_{j \in \mathrm{supp}(b) \setminus B} (1-c_j) \left( \sum_{A\subseteq \mathrm{supp}(a)} \prod\limits_{j \in A} c_j \  \prod\limits_{j \in \mathrm{supp}(a) \setminus A} (1-c_j) \right) \label{ineq:paran}
\end{align}
We will next simplify the two summands in the above equation. Towards this, note that, in the current case, $c_i = a_i$ for all $i \in \mathrm{supp}(a)$ and $c_i = b_i$ for all $i \in \mathrm{supp}(b)$. Further, the sums $\sum_{A\subseteq \mathrm{supp}(a)} \prod\limits_{j \in A} a_j  \ \prod\limits_{j \in \mathrm{supp}(a) \setminus A} (1-a_j)$ and $\sum_{B\subseteq \mathrm{supp}(b)} \prod\limits_{j \in B} b_j \ \prod\limits_{j \in \mathrm{supp}(b) \setminus B} (1-b_j)$ are equal to $1$, since the sums correspond to the total probability of drawing a set from the product distributions induced by $a$ and $b$, respectively. These observations imply that both the parenthesized terms in equation (\ref{ineq:paran}) are equal to $1$. Hence, the equation reduces to 
 \begin{align*}
F(c) & \leq \sum_{A\subseteq \mathrm{supp}(a)}  f(A) \ \prod\limits_{j \in A} a_j \ \prod\limits_{j \in \mathrm{supp}(a) \setminus A} (1-a_j) \ \ \ + \ \ \ \sum_{B\subseteq \mathrm{supp}(b)} f(B) \ \prod\limits_{j \in B} b_j \ \prod\limits_{j \in \mathrm{supp}(b) \setminus B} (1-b_j)\\
& = F(a) + F(b).
\end{align*}
Here, the last equality follows from the definitions of $F(a)$ and $F(b)$. This completes the proof of the base case.

\noindent
\paragraph{Induction step:} $\mathcal{P}(k-1) \implies \mathcal{P}(k)$ for $1\leq k \leq m$. To prove $\mathcal{P}(k)$ we fix any $a,b,c \in [0,1]^m$ with the properties that $c = b+a$ and $k = |\mathrm{supp}(a) \cap \mathrm{supp}(b)| \geq 1$. Since $k \geq 1$, we know there exists a $j^* \in [m]$ such that $a_{j^*}, b_{j^*} > 0$. Define sets $S_a \coloneqq \mathrm{supp}(a) \setminus \{j^* \}$ and $S_b  \coloneqq \mathrm{supp}(b) \setminus \{j^* \}$.

Additionally, let $a^{-}, b^{-}, c^- \in [0,1]^m$ be the projected vectors obtained by setting the $j^*$th component to zero: set $a^{-}_k = a_k$,  $b^{-}_k = b_k$, and $c^{-}_k = c_k$, for all $k \in [m]\setminus \{j^*\}$, and set $a^-_{j^*} = b^-_{j^*} = c^-_{j^*} = 0$. Similarly, define $a^{+}, b^{+}, c^+ \in [0,1]^m$ to be the vectors obtained by setting the $j^*$th component to one: set $a^{+}_k = a_k$,  $b^{+}_k = b_k$, and $c^{+}_k = c_k$, for all $k \in [m]\setminus \{j^*\}$, and set $a^+_{j^*} = b^+_{j^*} = c^+_{j^*} = 1$. 
    
The definition of $F$ gives us 
    \begin{align*}
        F(a) & = a_{j^*} \ F(a^+) \ \ + \ \ (1-a_{j^*}) \ F(a^-) \\
        & = a_{j^*}  \  F(a^+) + (c_{j^*} - a_{j^*})  F(a^-) + (1-c_{j^*}) F(a^-) \numberthis \label{equation:MLE-subadditive-one}
    \end{align*}
 Similarly, 
    \begin{align*}
        F(b) & = b_{j^*} \  F(b^+) + (1-b_{j^*})  F(b^-)\\
        & = b_{j^*}  \ F(b^+) + (c_{j^*}-b_{j^*})  F(b^-) + (1-c_{j^*}) F(b^-)\\
        & = (c_{j^*} - a_{j^*})  F(b^+) + a_{j^*} \  F(b^-) + (1-c_{j^*}) F(b^-) \numberthis \label{equation:MLE-subadditive-two}
    \end{align*}
The last equality follows from $c_{j^*} = b_{j^*} + a_{j^*}$. Summing equations~(\ref{equation:MLE-subadditive-one}) and~(\ref{equation:MLE-subadditive-two}) gives us
    \begin{align}
        F(a) + F(b) = a_{j^*} \left(F(a^+) + F(b^-)\right) + (c_{j^*} - a_{j^*}) \left(F(a^-) + F(b^+)\right) + (1-c_{j^*}) \left(F(a^-) + F(b^-)\right) \label{equation:MLE-subadditive-LHS}
    \end{align}
Further, we can write $F(c)$ 
   	\begin{align}
  		F(c) = c_{j^*} \ F(c^+) + (1-c_{j^*}) F(c^-) = a_{j^*} F(c^+) + (c_{j^*} - a_{j^*}) F(c^+) + (1-c_{j^*}) F(c^-) \label{equation:MLE-subadditive-RHS}
   	\end{align}
   	We will now prove the following three inequalities: 
	
\noindent	
(i) $F(c^+) \leq F(a^+) + F(b^-)$, 

\noindent	
(ii) $F(c^+) \leq F(a^-) + F(b^+)$, and 

\noindent	
(iii) $F(c^-) \leq F(a^-) + F(b^-)$. 

Note that, the desired inequality ($F(c) \leq F(a) + F(b)$) follows from (i), (ii), and (iii) along with a term-by-term comparison between equations (\ref{equation:MLE-subadditive-LHS}) and (\ref{equation:MLE-subadditive-RHS}). 
To establish inequality (i), note that $c^+ = a^+ + b^-$ and $|\mathrm{supp}(a^+) \cap \mathrm{supp}(b^-)| = k - 1$. Hence, the induction hypothesis (i.e., proposition $\mathcal{P}(k-1)$) implies that $F(c^+) \leq F(a^+) + F(b^-)$. Inequalities (ii) and (iii) follow from similar arguments. This completes the induction step and the theorem stands proved.
\end{proof}

We also include an elegant, coupling-based argument for \Cref{lemma:subadditive-mle} that was suggested by an anonymous reviewer. For each element $e \in [m]$, sample a value $u_e$ uniformly from $[0,1]$ and independently from all other samples. Initialize sets $A,B = \emptyset$. Put the element $e$ in the set $A$ if $u_e \leq a_e$ and put $e$ in the set $B$ if $a_e < u_e \leq a_e + b_e$. Then, by construction, the random sets $A$ and $B$ are always disjoint, and $A$ is distributed according to the \emph{multilinear distribution}, where each element $e$ is placed into $A$ independently with probability $a_e$. Similarly, $B$ is distributed according to the multilinear distribution corresponding to $b$. Moreover, $A \cup B$ is distributed according to the multilinear distribution corresponding to $c=a+b$. By $\sigma$-subadditivity, we have $f(A)+f(B) \geq f(A \cup B)$, taking expectation over both sides gives us $F(a)+F(b) \geq F(c)$.

\subsection{Envy-Free Division of Burnt Cake with $\sigma$-Subadditive Valuations} 
\label{section:ef-cake-subadd}
This section establishes the existence of envy-free cake divisions under valuations that are $\sigma$-subadditive and globally nonnegative.
 
Recall that a cake valuation $f$ is said to be $\sigma$-subadditive if $f( I \cup J) \leq f(I) + f(J)$ for all disjoint intervals $I, J \subseteq [0,1]$. More generally, $f$ is said to be $c\sigma$-subadditive if the inequality holds for all contiguous intervals $I = [x,y]$ and $J = [y,z]$, with $0 \leq x \leq y \leq z \leq 1$. Further, $f$ is globally nonnegative if $f([0,1]) \geq 0$; indeed, parts of the cake can still have negative value (can be burnt) under a globally nonnegative $f$. Also, recall that, throughout the paper, we will assume that the agents' valuations are normalized, $f_i(\emptyset) = 0$.
  
The theorem below is established by applying Sperner's lemma along the lines of \cite{Su1999rental}. However, the proof here is distinct in its utilization of $c\sigma$-subadditivity.  

\begin{theorem}\label{theorem:burnt-cake-envy-free}
Let $\mathcal{C} = \langle [n], \{f_i\}_{i=1}^n \rangle$ be a cake division instance with continuous, $c\sigma$-subadditive, and globally nonnegative valuations. Then, $\mathcal{C}$  necessarily admits a contiguous envy-free division $\mathcal{I} = (I_1, \ldots, I_n)$.
\end{theorem}
\begin{proof}
We will establish the theorem by applying Sperner's lemma. Consider the $(n-1)$-simplex $\Delta^{n-1} \coloneqq \{ x \in \mathbb{R}_{\geq 0}^n \mid \sum_{i=1}^n x_i = 1\}$. Note that $\Delta^{n-1}  = \mathrm{conv}(\{e_1, \ldots, e_n\})$, where $e_1, \ldots, e_n$ are the standard basis vectors of $\mathbb{R}^n$. Further, each point $x = (x_1, \ldots, x_n) \in \Delta^{n-1}$ induces a division of the cake, $[0,1]$, into $n$ intervals, with $[\sum_{j=0}^{k-1} x_j, \sum_{j=0}^{k} x_j]$ as the $k${th} interval from left; here, by convention, $x_0 =0$. We will write $\calI^x=(I^x_1, \ldots, I^x_n)$ to denote the contiguous cake division induced by $x$, i.e., interval $I^x_k = [\sum_{j=0}^{k-1} x_j, \sum_{j=0}^{k} x_j]$ and $\len(I^x_k) = x_k$. 

We denote the $n$ facets of $\Delta^{n-1}$ as $F_1,\ldots, F_n$ and index them such that $i$th facet $  F_i \coloneqq \mathrm{conv}(\{e_1, \ldots, e_{i-1}, \allowbreak e_{i+1}, \ldots, e_n \}) = \{x \in \Delta^{n-1} : x_i = 0\}$. As noted in \cite{Su1999rental}, by repeatedly performing the Barycentric subdivision, the simplex $\Delta^{n-1}$ can be triangulated to obtain a triangulation $\mathcal{T}$ with arbitrary small granularity. In addition, the Barycentric subdivision ensures that each vertex $v \in \mathrm{vert}(\mathcal{T})$ can be mapped to an agent $k \in [n]$, which is referred to as the \emph{owner} of the vertex $v$. In particular, this ownership mapping satisfies the following property: in each {elementary simplex} of the triangulation $\mathcal{T}$, the $n$ vertices are owned by $n$ different agents. That is, each agent owns exactly one vertex in every elementary simplex of $\mathcal{T}$~\cite{Su1999rental}. 
 
We next define a labeling $\lambda: \mathrm{vert}(\mathcal{T}) \mapsto [n]$ and then apply the Sperner's lemma. Towards this, we consider each vertex $y \in \mathrm{vert}(\mathcal{T})$ and the agent $i \in [n]$ who owns $y$. Then, we set $\lambda(y)$ as the index of the interval most-valued by agent $i$ in the induced division $\calI^{y}=(I^{y}_1,\ldots, I^{y}_n)$. Specifically, $\lambda(y) \in \argmax_{k \in [n]} \ \big\{f_i(I^{y}_k) \big\}$. In case of a tie (i.e., if $\left| \argmax_{k \in [n]} \ \big\{f_i(I^{y}_k) \big\} \right| > 1$), the tiebreak is performed in favor of a nonempty interval, i.e., $\lambda(y) = \widehat{k}$ with $\widehat{k} \in \argmax_{k \in [n]} \ \big\{f_i(I^{y}_k) \big\}$ and $y_{\widehat{k}} = \len \left(I^y_{\widehat{k}} \right) > 0$. We will next show that, under $\sigma$-subadditive valuations, one can always execute the tie-breaking rule and the resulting labeling $\lambda$ is in fact a Sperner labeling. 

Towards this, we write $Z_y \coloneqq \{ j \in [m] \mid y_j = 0\}$ to denote the set of zero component of vertex $y$. Note that, for each index $j \in Z_y$, the $j$th interval from the left in the induced division (i.e., $I^y_j$) is empty: $\len(I^y_j) = y_j = 0$. Further, $Z_y$ is the collection of indices of all the facets that contain $y$; in particular, $y \in \cap_{j \in Z_y} F_j$. 
Also, $[n] \setminus Z_y$ corresponds to the set of nonzero components in $y$. We will show (in the claim below) that $ \left( [n] \setminus Z_y \right) \cap \left(\argmax_{k \in [n]} \ \big\{f_i(I^{y}_k) \big\} \right) \neq \emptyset$. This intersection implies that stated tie-breaking rule can always be executed. 

\begin{claim}
\label{claim:stronghungry}
For every vertex $y$, we have $\left( [n] \setminus Z_y \right) \cap \left(\argmax_{k \in [n]} \ \big\{f_i(I^{y}_k) \big\} \right) \neq \emptyset$.
\end{claim}
\begin{proof}
For the agent $i \in [n]$ who owns $y$ we have 
\begin{align*}
0 & \leq f_i([0,1]) \tag{$f_i$ is globally nonnegative}\\
& = f_i(\cup_{k=1}^n I^y_k) \tag{$\calI^y= (I^y_1, \ldots, I^y_n)$ is a cake division}\\
& \leq \sum_{k=1}^n f_i(I^y_k) \tag{$f_i$ is $c\sigma$-subadditive}\\
& = \sum_{k \in [n]\setminus Z_{y}} f_i(I^y_k) \tag{$f_i(I^y_j) = f_i(\emptyset) = 0$ for all $j \in Z_y$}\\
& \leq |[n] \setminus  Z_{y}|  \ \max_{k \in [n]\setminus Z_{y}} f_i(I^y_k).
\end{align*}
Hence, there always exists an index $k^* \in [n] \setminus Z_y$ with the property that $f_i(I^y_{k^*}) \geq 0 = f_i(\emptyset)$. That is, we always have a nonempty interval $I^y_{k^*}$ of value at least as much as an empty one. This observation implies that $\left( [n] \setminus Z_y \right) \cap \left(\argmax_{k \in [n]} \ \big\{f_i(I^{y}_k) \big\} \right) \neq \emptyset$ and completes the proof of the claim. 
\end{proof}

Claim \ref{claim:stronghungry} ensures that, while defining the labeling $\lambda$, the tie-breaking rule (in favor of non-empty intervals) can always be implemented. In particular, $\lambda$ is well-defined. Moreover, the resulting property that, for every vertex $y$, the label $\lambda(y) \in [n] \setminus Z_y$ implies that $\lambda$ is a Sperner labeling. Therefore, via \Cref{theorem:brouwer-fpt}, we obtain that exists a fully-labeled elementary simplex $L$ in the triangulation $\mathcal{T}$.
 
Finally, using a limiting argument and the continuity of the cake valuations $f_i$, we will prove the existence of a point $x^* \in \Delta^{n-1}$ with the property that the induced division $\calI^{x^*}$ is envy-free. As established above, in any Barycentric triangulation $\mathcal{T}$ there exists a fully-labeled elementary simplex $L$. By performing finer and finer the Barycentric subdivisions of $\Delta^{n-1}$,\footnote{In particular, given a triangulation $\mathcal{T}$, to obtain a finer triangulation, we can further triangulate each elementary simplex of $\mathcal{T}$ via a Baycentric subdivision.} we can obtain an infinite sequence $L_1, L_2, \ldots$ of fully labeled elementary simplices that converge at a point $x^* \in \Delta^{n-1}$. We will show that $\calI^{x^*}$ is an envy-free cake division. Towards this, note that there are $n!$ possible fully labeled elementary simplices, based on the number of ways to label the $n$ vertices of an $(n-1)$-simplex with distinct labels. Hence, applying the Bolzano-Weierstrass theorem, we get that there must exist a converging (to $x^*$) subsequence $L_{i_1}, L_{i_2}, \ldots$ such that the labeling of each $L_{i_j}$ is the same. Since the agents' valuations $f_i$ are continuous, we get that at the limit point $x^*$, all the agents will assign labels consistent with the one at $\mathcal{L}_{i_j}$s. In particular, we obtain that under the cake division $\calI^{x^*}$, all the agents will have distinct most-valued intervals. That is, the interval  chosen by each agent will be of value at least as much as any other interval in $\calI^{x^*}$. Therefore, the contiguous cake division $\calI^{x^*}$  is indeed envy-free. This completes the proof of the theorem. 
\end{proof}

\subsection{$\EFt$ Allocations Under $\sigma$-Subadditive Valuations}
\label{section:ef-three-subadditive}
This section addresses allocation of indivisible items under $\sigma$-subadditive valuations. Here, we establish that $\EFt$ allocations always exist for $\sigma$-subadditive and globally nonnegative valuations. Recall that a set function $v$ is said to be globally nonnegative if $v([m]) \geq 0$. 

The $\EFt$ existential guarantee is obtained by first constructing a cake division instance $\mathcal{C}$---via \textsc{CakeConstruction} (\Cref{algo:cake-reduction})---from the given discrete fair division instance $\mathcal{D}$. Then, we invoke Theorems \ref{lemma:subadditive-mle} and \ref{theorem:burnt-cake-envy-free} to argue that $\mathcal{C}$ admits an envy-free division $\mathcal{I}$. We will show in the proof of the following theorem that rounding $\calI$---via \textsc{CakeRounding} (\Cref{algo:cake-rounding})---gives us an $\EFt$ allocation $\calA$.  

\begin{theorem}[$\EFt$ for $\sigma$-Subadditive]\label{theorem:ef-three-subadditive}
Every discrete fair division instance $\langle [n], [m], \{ v_i \}_{i=1}^n \rangle$, with $\sigma$-subadditive, globally nonnegative, and normalized valuations, admits an $\EFt$ allocation $\calA$; in particular, under $\alloc = (A_1, \ldots, A_n)$, for all $i,j \in [n]$, there exist subsets $A'_i, A'_j$ with the properties that $v_i(A'_i) \geq v_i(A'_j)$ and $|A_i \triangle A'_i| + |A_j \triangle A'_j| \leq 3$.  
\end{theorem}
\begin{proof} 
Write $\calD = \langle [n], [m], \{ v_i \}_{i=1}^n \rangle$ to denote the given fair division instance and let $\calC = \langle [n], \{ f_i \}_{i=1}^n \rangle$ be the cake division instance obtained by executing \textsc{CakeConstruction} (\Cref{algo:cake-reduction}) with $\calD$ as input. By construction, the cake valuations $f_i$ are continuous. We will next show that here $f_i$s are $c\sigma$-subadditive as well. 

Towards this, consider any two contiguous intervals $I_1, I_2 \subseteq [0,1]$; in particular, $I_1$ and $I_2$ are disjoint and their union $I \coloneqq  I_1 \cup I_2$ is itself an interval. Note that 
\begin{align}
f_i(I) = V_i \circ b (I) = V_i \circ b (I_1 \cup I_2) = V_i (b(I_1) + b(I_2)) \label{ineq:multi-fV}
\end{align}
Here, the last equality follows from the fact that the function $b$ (defined in \Cref{section:cake-reduction}) is additive. Write vector $a = b(I_1) \in [0,1]^m$ along with $b= b(I_2)$ and $c = b(I_1 \cup I_2)$. Since the intervals $I_1$ and $I_2$ are disjoint, vectors $a,b,c \in [0,1]^m$ satisfy $c = a+b$. Further, given that $V_i$ is the multilinear extension of $\sigma$-subadditive valuation $v_i$, Theorem \ref{lemma:subadditive-mle} holds for $V_i$ and we have $V_i(c) \leq V_i(a) + V_i(b)$. That is, $ V_i (b(I_1) + b(I_2)) \leq V_i (b(I_1)) + V_i (b(I_2))$. This inequality and equation (\ref{ineq:multi-fV}) give us 
\begin{align}
f_i(I) = V_i (b(I_1) + b(I_2))  \leq V_i (b(I_1)) + V_i (b(I_2)) = f_i(I_1) + f_i(I_2) \label{ineq:f-is-suba}
\end{align}
Equation (\ref{ineq:f-is-suba}) shows that the constructed cake valuations $f_i$s are $c\sigma$-subadditive. In addition, note that the global nonnegativity of $v_i$ implies that $f_i$ is globally nonnegative as well, $f_i([0,1]) = V_i(b([0,1])) = v_i([m]) \geq 0$. Hence, applying Theorem \ref{theorem:burnt-cake-envy-free}, we obtain that the constructed cake division instance $\calC = \langle [n], \{ f_i \}_{i=1}^n \rangle$ admits a contiguous envy-free division $\calI$.  

Finally, we invoke \textsc{CakeRounding} (\Cref{algo:cake-rounding}) over instances $\calD$ and $\calC$ along with the division $\calI$ to obtain allocation $\calA$. Lemma \ref{lemma:cake-rounding} (instantiated with $\alpha =0$) implies that, as desired, the allocation $\calA$ is $\EFt$ for the given discrete fair division instance $\calD$. This completes the proof of the theorem. 
\end{proof}

\section{Envy-Freeness Under Nonnegative Valuations} 
\label{section:fair-div-nn}
This section develops envy-freeness guarantees for valuations that are nonnegative (and not necessarily monotone). Specifically, we will focus on the indivisible-items and prove that 
 any discrete fair division instance, in which the agents' valuations are nonnegative admits an $\EFt$ allocation.

\subsection{$\EFt$ Under Nonnegative Valuations} 
\label{section:eft-nn} 
Our $\EFt$ existential result for nonnegative valuations is established via a constructive proof -- we prove that \Cref{algo:EFThree-nonnegative} (detailed below) necessarily finds such allocations. Indeed, the algorithm is not claimed to be polynomial time since it entails computing a contiguous envy-free division, which is known to be a \textsc{PPAD}-complete problem, even for a constant number of agents~\cite{gao2024hardness}. 

\Cref{algo:EFThree-nonnegative} first takes the input instance $\langle [n], [m], \{v_i\}_{i=1}^n \rangle$ and constructs valuations $v^{\varepsilon}_i(\cdot)$s. This affine transformation essentially ensures that we can invoke the hungry condition in the cake division instance constructed via \textsc{CakeConstruction} (\Cref{algo:cake-reduction}). We will also show, in the algorithm's analysis, that allocation $\calA$ obtained via  \textsc{CakeRounding} (\Cref{algo:cake-rounding}) is $\EFt$, not only under the modified valuations $v^{\varepsilon}_i(\cdot)$, but also under the given ones $v_i$. 

\floatname{algorithm}{Algorithm}
\begin{algorithm}[h]
\caption{\textsc{ConstructiveEFE3}} \label{algo:EFThree-nonnegative}
\begin{tabularx}{\textwidth}{X}
{\bf Input:} A discrete fair division instance $\mathcal{D} = \langle [n], [m], \{ v_i \}_{i=1}^n \rangle$. \\
{\bf Output:} An $\EFt$ allocation $\alloc = (A_1, \ldots, A_n)$.
\end{tabularx}
  \begin{algorithmic}[1]
  		\STATE Define $\delta \coloneqq \min\limits_{i\in[n]} \ \min\limits_{S, T \subseteq [m]} \ \Big\{v_i(S) - v_i(T) \mid v_i(S) - v_i(T) > 0 \Big\}$ and set  parameter $\varepsilon \coloneqq \frac{\delta}{2} >0$. \label{line:ef-three-delta-def}
  		\STATE For each agent $i \in [n]$, define $v_i^\varepsilon : 2^{m} \rightarrow \mathbb{R}_{\geq 0}$ as follows: $v_i^\varepsilon(S) = v_i(S) + \varepsilon$, for every (nonempty) subset $S \neq \emptyset$, and $v_i^\varepsilon(\emptyset) = v_i(\emptyset) = 0$. \label{line:ef-three-eps-function-def}
		\STATE With instance $\mathcal{D}^\varepsilon \coloneqq \langle [n], [m], \{ v^\varepsilon_i \}_{i=1}^n \rangle$ as input, invoke \Cref{algo:cake-reduction}: \\ $\langle [n], \{ f_i \}_{i=1}^n \rangle = \textsc{CakeConstruction} (\mathcal{D}^\varepsilon)$. \label{line:ef-three-cake-reduction}	
		\STATE Let $\mathcal{I} = (I_1, \ldots, I_n)$ be a contiguous envy-free cake division in the constructed instance $\mathcal{C}^\varepsilon \coloneqq \langle [n], \{ f_i \}_{i=1}^n \rangle$. \label{line:ef-divisible-alloc}  \COMMENT{We will show that such a division $\calI$ always exists.}
		\RETURN Allocation $\calA = \textsc{CakeRounding}( \mathcal{D}^\varepsilon, \mathcal{C}^\varepsilon, \mathcal{I})$.\label{line:ef-three-cake-rounding}
		\end{algorithmic}
\end{algorithm}

\begin{theorem}[$\EFt$ for Nonnegative]\label{theorem:ef-three-nonnegative}
Given as input any discrete fair division instance $\langle [n], [m], \{ v_i \}_{i=1}^n \rangle$, with nonnegative and normalized valuations,~\Cref{algo:EFThree-nonnegative} always returns an $\EFt$ allocation $\alloc = (A_1, \ldots, A_n)$. That is, under nonnegative and normalized valuations, there always exists an allocation $\alloc = (A_1, \ldots, A_n)$ in which, for all $i,j \in [n]$, there exist subsets $A'_i, A'_j$ with the properties that $v_i(A'_i) \geq v_i(A'_j)$ and $|A_i \triangle A'_i| + |A_j \triangle A'_j| \leq 3$. 
 \end{theorem}
\begin{proof}
\Cref{algo:EFThree-nonnegative} begins by identifying $\delta$ as the minimum positive difference between an agent's values across the subsets, $\delta \coloneqq \min\limits_{i\in[n]} \ \min\limits_{S, T \subseteq [m]} \ \Big\{v_i(S) - v_i(T) \mid v_i(S) - v_i(T) > 0 \Big\}$. We can assume, without loss of generality, that for each nonnegative valuation $v_i$, there exists a subset $X$ such that $v_i(X)>0$; an agent that values all the subsets at $0$ can be removed from consideration since such an agent will always be envy-free. Hence, the parameter $\delta$ is well-defined and so is $\varepsilon = \frac{\delta}{2} > 0$. 

Then, from the given discrete fair division instance $\mathcal{D} = \langle [n], [m], \{ v_i \}_{i=1}^n \rangle$, the algorithm constructs instance 
$\mathcal{D}^\varepsilon \coloneqq \langle [n], [m], \{ v^\varepsilon_i \}_{i=1}^n \rangle$   by defining set functions $v^\varepsilon_i$ as follows:  $v_i^\varepsilon(S) = v_i(S) + \varepsilon$, for each (nonempty) subset $S \neq \emptyset$, and $v_i^\varepsilon(\emptyset) = v_i(\emptyset) = 0$. Since $v_i$s are normalized and nonnegative, each $v_i^\varepsilon$ is normalized and has a positive value for every nonempty subset: $v^\varepsilon_i(S) \geq \varepsilon > 0$ for every $S \neq \emptyset$.  

First, we will show that the cake division instance $\mathcal{C}^\varepsilon = \langle [n], \{ f_i \}_{i=1}^n \rangle$ (obtained in Line \ref{line:ef-three-cake-reduction} via \textsc{CakeConstruction}) necessarily admits an envy-free division $\calI=(I_1, \ldots, I_n)$. Then, we will prove that the rounded allocation $\calA$ (computed in Line \ref{line:ef-three-cake-rounding} via \textsc{CakeRounding}) is $\EFt$ for $\mathcal{D}^\varepsilon$ as well as $\mathcal{D}$. 
	
To show that $\mathcal{C}^\varepsilon$ admits a contiguous envy-free cake division, $\calI$, we will invoke \Cref{theorem:cake-sperner-ef}. The theorem requires the cake valuations $\{ f_i \}_{i=1}^n$ to be continuous and uphold the hungry condition. Recall that, $f_i = V_i \circ b$, where $V_i$ is the multilinear extension of $v^\varepsilon_i$ and the linear function $b$ is as defined in \Cref{section:cake-reduction}. Hence, here, the $f_i$s are continuous. For the hungry condition, consider any nonempty interval $J\subseteq [0,1]$ (i.e., $\len(J) > 0$). For such an interval $J$, note that $b(J) \neq \mathbf{0}$. Further, given that  $V_i$ is the multilinear extension of $v^\varepsilon_i$ and $v^\varepsilon_i(S) > 0$, for every nonempty subset $S$, we have $f_i(J) = V_i(b(J)) > 0 = f_i(\emptyset)$. That is, the hungry condition holds and, hence, \Cref{theorem:cake-sperner-ef} ensures the existence of a contiguous cake division $\mathcal{I}=(I_1, \ldots, I_n)$ that satisfies envy-freeness:   
\begin{align}\label{equation:EFThree-nonnegative-envy-free-cake}
f_i(I_i) \geq f_i(I_j) \ \ \text{ for all agents } i,j \in [n].
\end{align}

Finally, we will prove that that the returned allocation $\calA$ (obtained by rounding $\calI$ via \textsc{CakeRounding}) is $\EFt$ for $\mathcal{D}^\varepsilon$ and $\mathcal{D}$. Using equation (\ref{equation:EFThree-nonnegative-envy-free-cake}) and invoking \Cref{lemma:cake-rounding}---with instances $\mathcal{D}^\varepsilon$ and $\mathcal{C}^\varepsilon$ along with division $\calI$ and parameter $\alpha = 0$---we obtain that $\calA$ is $\EFt$ for $\mathcal{D}^\varepsilon$. That is, in $\calA=(A_1, \ldots, A_n)$, for all $i,j \in [n]$, there exists subsets $A'_i, A'_j \subseteq [m]$ such that $v^\varepsilon_i(A'_i) \geq v^\varepsilon(A'_j)$ and $|A_i \triangle A'_i| + |A_j \triangle A'_j| \leq 3$. 

To complete the proof we next establish that $\calA$ is in fact $\EFt$ for the given instance $\mathcal{D} = \langle [n], [m], \{ v_i \}_{i=1}^n \rangle$. Towards this,  we will show that the inequality $v^\varepsilon_i(A'_i) \geq v_i^\varepsilon(A'_j)$ implies $v_i(A'_i) \geq v_i(A'_j)$.  Assume, towards a contradiction, that instead $v_i(A'_i) < v_i(A'_j)$. Hence, 
 \begin{align*}
0 & < v_i(A'_j) - v_i(A'_i) \\
 & \leq v^\varepsilon_i(A'_j)  - v_i(A'_i) \tag{$v_i(S) \leq v_i^\varepsilon(S)$ for all $S\subseteq [m]$} \\
 & \leq v^\varepsilon_i(A'_j) - \left( v_i^\varepsilon(A'_i) - \varepsilon \right) \tag{$v_i^\varepsilon(S) \leq v_i(S) + \varepsilon$ for all $S\subseteq [m]$}\\ 
 & \leq \varepsilon \tag{since $v^\varepsilon_i(A'_j) - v^\varepsilon_i(A'_i) \leq 0$}\\
 & < \delta,
\end{align*}
where the final equality follows from $\varepsilon<\delta$. However, this bound, $0 < v_i(A'_j) - v_i(A'_i) < \delta$, contradicts the definition of $\delta$. Therefore, for all $i,j \in [n]$, we have $v_i(A'_i) \geq v_i(A'_j)$ where $A'_i,A'_j \subseteq [m]$ satisfy $|A_i \triangle A'_i| + |A_j \triangle A'_j| \leq 3$. That is, $\alloc = (A_1, \ldots, A_n)$ is $\EFt$ in the given instance $\mathcal{D}$.  The theorem stands proved. 
\end{proof}

\section{Equitability Under Nonnegative Valuations}
\label{section:eq-nn}

This section develops our existential guarantees for equitability under nonnegative valuations. Section \ref{section:equitable-cake-div} addresses cake division, and Section \ref{section:eqt-nn} covers the indivisible-items setting.

\subsection{Contiguous Equitable Cake Division Under Nonnegative Valuations}
\label{section:equitable-cake-div}
This section addresses equitable cake divisions under nonnegative valuations. The work of Bhaskar, Sricharan, and Vaish \cite{bhaskar2025connected} shows that, under nonnegative cake valuations, equitability is guaranteed along with a useful ordering property: for any ordering (indexing) of the agents, $\{1,2,\ldots, n\}$, there exists a contiguous equitable division $\calI = (I_1, I_2, \ldots, I_n)$ that allocates the cake---left to right---in that order. \cite{bhaskar2025connected} obtains this result via Sperner's Lemma; their existential result holds under more general conditions. Next, for equitable cake division, we provide a complementary proof of existence via Brouwer's fixed point theorem. We will later leverage the ordering property obtained here to develop our algorithms for finding approximately equitable divisions.

\begin{theorem}\label{theorem:eq-cake-nn}
Let $\calC = \langle [n], \{f_i\}_{i=1}^n \rangle$ be a cake division instance with continuous, nonnegative, and normalized valuations. Then, for any (re)ordering of the agents, $\{1,2,\ldots, n\}$, there exists a contiguous cake division $\calI = (I_1, \ldots, I_n)$ in $\calC$ such that, for all $1 \leq i < n$, equitability holds, $f_i(I_i) = f_{i+1}(I_{i+1})$, and the interval $I_i$ lies to the left of $I_{i+1}$. 
\end{theorem}
\begin{proof}
We will establish this result via Brouwer's fixed point theorem (\Cref{theorem:brouwer-fpt}). Write $\Delta^{n-1} \coloneqq \{x \in [0,1]^n \mid \sum_{i=1}^n x_i = 1\}$ to denote the $(n-1)$-simplex. With every point $x \in \Delta^{n-1}$ we can associate a contiguous division $\calI^x = (I^x_1, \ldots, I^x_n)$ of the the cake $[0,1]$ by setting the $k$th interval from the left $I^x_k \coloneqq \left[ \sum_{j=0}^{k-1} x_j, \sum_{j=0}^{k} x_j\right]$; here, by convention, $x_0 =0$. Note that $\len(I^x_k) = x_k$. Also, let $\gamma > 0$ be a sufficiently large parameter such that $f_i([a,b]) \leq \gamma \  (b-a)$ for each $i \in [n]$ and all intervals $[a,b] \subseteq [0,1]$. The normalization and (Lipschitz) continuity of $f_i$s ensure that such a parameter $\gamma < \infty$ exists. 

To apply Brouwer's fixed point theorem, we define a function $g: \Delta^{n-1} \mapsto \Delta^{n-1}$. Specifically, for any $x \in \Delta^{n-1}$, write $g(x) = (g_1(x), g_2(x), \ldots, g_n(x))$ with the $i$th component 
\begin{align*}
g_i(x) \coloneqq  x_i + \frac{1}{\gamma} \left(\frac{1}{n} \sum_{j=1}^n f_j(I^x_j) - f_i(I^x_i)\right).
\end{align*}
The fact that $g_i$s are obtained by composing continuous valuations $f_i$s implies that $g$ is continuous as well. Next, we will show that, for each $x \in \Delta^{n-1}$, the image $g(x) \in \Delta^{n-1}$. Consider the sum of the components of $g(x)$:

\begingroup 
\allowdisplaybreaks
\begin{align*}
\sum_{i=1}^n g_i(x) & = \sum_{i=1}^n \left( x_i + \frac{1}{\gamma} \left(\frac{1}{n} \sum_{j=1}^n f_j(I^x_j) - f_i(I^x_i)\right) \right) \\
& = \sum_{i=1}^n x_i + \frac{1}{\gamma} \sum_{i=1}^n \left( \frac{1}{n} \sum_{j=1}^n f_j(I^x_j) \right)- \frac{1}{\gamma}\sum_{i=1}^nf_i(I^x_i)  \\
& = \sum_{i=1}^n x_i + \frac{1}{\gamma}   \frac{n}{n} \sum_{j=1}^n f_j(I^x_j) - \frac{1}{\gamma}\sum_{i=1}^nf_i(I^x_i)  \\
& = 1 + \frac{1}{\gamma}\sum_{j=1}^n f_j(I^x_j) - \frac{1}{\gamma}\sum_{i=1}^n f_i(I^x_i) \tag{$\sum_{i=1}^n x_i = 1$, since $x \in \Delta^{n-1}$}\\
& = 1.
\end{align*}
\endgroup

Further, the following inequality shows that $g_i(x) \geq 0$ for each $i \in [n]$ and $x \in \Delta^{n-1}$: 
\begin{align*}
g_i(x) & = x_i + \frac{1}{\gamma} \left(\frac{1}{n} \sum_{j=1}^n f_j(I_j) - f_i(I_i)\right)\\
& \geq x_i - \frac{1}{\gamma} f_i(I_i) \tag{nonnegativity of $f_j$s}\\
& \geq 0.
\end{align*}
Here, the last inequality follows from the definition of $\gamma$; in particular, $f_i(I_i) \leq \gamma \len(I_i) = \gamma x_i$. Hence, the vector $g(x) \in \Delta^{n-1}$. The above-mentioned observations imply that that the continuous function $g$ maps points in the nonempty, convex, compact set $\Delta^{n-1}$ to $\Delta^{n-1}$ itself. Therefore, Brouwer's fixed point theorem (\Cref{theorem:brouwer-fpt}) implies that there exists $x^* \in \Delta^{n-1}$ with the property that $g(x^*) = x^*$. 

Write $\calI^*=(I^*_1, \ldots, I^*_n)$ to denote the cake division induced by $x^*$, i.e., $I^*_k =  [\sum_{j=0}^{k-1} x^*_j, \sum_{j=0}^i x^*_j]$ is the $k$th interval from the left in $\calI^{*}$. The equality $g(x^*) = x^*$ implies that, for each $i \in [n]$, we have $x^*_i = g_i(x^*) = x^*_i + \frac{1}{\gamma} \left(\frac{1}{n} \sum_{j=1}^n f_j(I^{*}_j) - f_i(I^{*}_i) \right)$. Simplifying we obtain $f_i(I^{*}_i) = \frac{1}{n} \sum_{j=1}^n f_j(I^{*}_j) = f_k(I^{*}_k)$ for all $i,k \in [n]$. Hence, the contiguous cake division $\calI^{*}=(I^*_1, \ldots, I^*_n)$ is equitable. Also, note that the interval $I^*_i$ lies to the left of $I^*_{i+1}$ for each $1 \leq i < n$. The theorem stands proved. 
\end{proof}

\subsection{$\EQt$ Allocations for Nonnegative Valuations}
\label{section:eqt-nn}
In this section, we establish the existence of $\EQt$ allocations for nonnegative valuations.  

The $\EQt$ existential guarantee for nonnegative valuations is obtained by first constructing a cake division instance $\mathcal{C}$---via \textsc{CakeConstruction} (\Cref{algo:cake-reduction})---from the given discrete fair division instance $\mathcal{D}$. Then, we invoke Theorem \ref{theorem:eq-cake-nn} to argue that $\mathcal{C}$ admits an equitable division $\mathcal{I}$. We will show in the proof of the following theorem that rounding $\calI$ yields an $\EQt$ allocation $\calA$ (\Cref{lemma:cake-rounding-2}).  

\begin{theorem}[$\EQt$ for Nonnegative]\label{theorem:eq-three-nonnegative}
Every discrete fair division instance $\langle [n], [m], \{ v_i \}_{i=1}^n \rangle$, with nonnegative and normalized valuations valuations, admits an $\EQt$ allocation $\calA$; in particular, under $\alloc = (A_1, \ldots, A_n)$, for all $i,j \in [n]$, there exist subsets $A'_i, A'_j$ with the properties that $v_i(A'_i) \geq v_j(A'_j)$ and $|A_i \triangle A'_i| + |A_j \triangle A'_j| \leq 3$. 
\end{theorem}
\begin{proof}
Write $\calD = \langle [n], [m], \{ v_i \}_{i=1}^n \rangle$ to denote the given fair division instance and let $\calC = \langle [n], \{ f_i \}_{i=1}^n \rangle$ be the cake division instance obtained by executing \textsc{CakeConstruction} (\Cref{algo:cake-reduction}) with $\calD$ as input. By construction, the cake valuations $f_i$ are continuous. We will next note that here $f_i$s are nonnegative as well. 

Recall that, for each interval $I \subseteq [0,1]$ the cake valuation $f_i(I) = V_i \circ b (I)$. Further, by definition of the linear function $b$, we have $b(I) \in [0,1]^m$. Now, the fact that the valuations $v_i$ are nonnegative imply that their multilinear extensions $V_i$ are nonnegative as well. Hence, we obtain that $f_i$s are nonnegative. Therefore, applying Theorem \ref{theorem:eq-cake-nn}, we obtain that the constructed cake division instance $\calC = \langle [n], \{ f_i \}_{i=1}^n \rangle$ admits a contiguous equitable division $\calI$.  

Finally, we invoke \textsc{CakeRounding} (\Cref{algo:cake-rounding}) over instances $\calD$ and $\calC$ along with the division $\calI$ to obtain allocation $\calA$. Lemma \ref{lemma:cake-rounding-2} (instantiated with $\alpha =0$) implies that, as desired, the allocation $\calA$ is $\EQt$ for the given discrete fair division instance $\calD$. This completes the proof of the theorem. 
\end{proof}

\section{Finding Approximately Equitable Divisions Under Nonnegative Valuations} 
\label{section:computation}
This section develops algorithms for finding approximately equitable divisions under nonnegative valuations. We start with cake division and then address the indivisible-items setting.

The ordering property obtained in Theorem \ref{theorem:eq-cake-nn} ensures that, under nonnegative valuations, we can restrict attention to equitable cake divisions $\calI^*=(I^*_1, \ldots, I^*_n)$ in which, for each $i \in [n]$, agent $i$'s interval, $I^*_i$, is the $i$th one from the left. That is, $I^*_1$ is the left-most interval in the division and $I^*_n$ is the right most. Our algorithm leverages this property and grid searches for a threshold $\tau$ that is additively close to the agents' value in the equitable division, i.e., close to $\tau^* \coloneqq f_i(I^*_i)$. With a (guessed) $\tau$ in hand, the algorithm constructs, for each agent $i \in [n]$, a family of intervals, $\calF^\tau_i$, that contains intervals of value close to $\tau$ for agent $i$; the endpoints of the intervals in $\calF^\tau_i$ are discretized to ensure that the families are of polynomial size. Then, via a dynamic program, we search for whether there exist intervals $I_i \in \calF^\tau_i$ that partition the cake and uphold the ordering (i.e., $I_i$ is to the left of $I_{i+1}$). This dynamic program is detailed in Section \ref{subsection:interval-partitioning}.

We will prove that, with the right choice of approximation parameters, for at least one $\tau$ the dynamic program will succeed in finding the ordered intervals $I_i \in \calF^\tau_i$ that partition the cake. Further, any such cake division $(I_1,\ldots, I_n)$ is guarantee to be approximately equitable, since $I_i \in \calF^\tau_i$ ensures that each agent's value, $f_i(I_i)$, is close to $\tau$. Details of our algorithm for finding approximately equitable cake divisions are provided in \Cref{section:fptas-for-eq-div}. This algorithmic result requires that the nonnegative cake valuations $f_i$s be Lipschitz continuous. Specifically, we assume that the there exists a known Lipschitz constant $\gamma > 0$ such that $\left| f_i([\ell, r]) - f_i ([\ell', r']) \right| \leq \gamma \|(\ell, r) - (\ell', r') \|_1$ for all intervals $[\ell, r], [\ell', r'] \subseteq [0,1]$ and all $i \in [n]$. The algorithm operates in the standard Robertson-Webb query model and its running time is polynomial in $n$, $\gamma$, and $1/\varepsilon$, where $\varepsilon$ is the absolute approximation factor. Specifically, our algorithm finds a cake division $\calI=(I_1, \ldots, I_n)$ in which $\left| f_i(I_i) - f_j(I_j) \right| \leq \varepsilon$, for all agents $i, j \in [n]$ (Theorem \ref{theorem:fptas-for-equitable-cake-div}).  

\Cref{section:eq-three-nonnegative} develops our algorithmic result for finding approximately equitable allocations of indivisible items. At a high level, we use ideas similar to the ones in Sections \ref{section:cake-reduction} and \ref{section:CakeRounding} and construct a cake instance from the given discrete fair division instance. However, before applying the equitable cake division algorithm (from Section \ref{section:fptas-for-eq-div}) we modify the cake valuations by adding a piecewise linear function to them. The need for this transformation stems from the requirement that computed, approximately equitable allocation $\calA = (A_1, \ldots, A_n)$ should consist of nonempty bundles. This property of nonempty bundles has important implications for the graph-theoretic applications mentioned earlier. Also, analogous to the Lipschitz requirement, in the discrete fair division case we assume an upper bound $\Lambda >0$ on the marginal values of $v_i$s, i.e., for each item $g \in [m]$ and subset $S \subseteq [m]$ we have $|v_i(S \cup \{g\}) - v_i(S)| \leq \Lambda$. Our algorithm runs in time polynomial in $n$, $m$, and $\Lambda$. It computes an allocation $\calA = (A_1, \ldots, A_n)$ with nonempty bundles that satisfy $|v_i(A_i) - v_j(A_j)| \leq 5\Lambda + 1$, for all $i, j \in [n]$; see Theorem \ref{theorem:EQ3ApxComputation}.  
	
\subsection{Ordered Interval Selection}\label{subsection:interval-partitioning}
In this section, we address an ordered interval selection problem. For the problem, we provide a dynamic programming algorithm, which will be used as a subroutine in subsequent methods. 

In this interval selection problem, we are given $n$ (indexed) families of intervals $\mathcal{F}_1, \mathcal{F}_2, \ldots, \mathcal{F}_n$. Here, each $\calF_i$ consists of finitely many cake intervals $[p,q]$, with $0 \leq p \leq q \leq 1$. The objective is to determine whether there exist $n$ pairwise disjoint intervals $I_i \in \calF_i$ that partition the cake and uphold the ordering, i.e., $I_i$ is to the left of $I_{i+1}$ in the cake, for each $1 \leq i < n$. Note that the requirement that the intervals partition the cake ensures $\cup_{i=1}^n I_i = [0,1]$ and $I_i \cap I_j = \emptyset$ for all $i \neq j$. If such intervals exist, then we want to efficiently find them as well. 

The following theorem shows that the \textsc{IntervalSelect} algorithm (\Cref{algo:interval-dp}) stated below solves this problem in polynomial time.  

\floatname{algorithm}{Algorithm}
\begin{algorithm}[h]
\caption{\textsc{IntervalSelect}} \label{algo:interval-dp}
\begin{tabularx}{\textwidth}{X}
{\bf Input:} Families of intervals $\mathcal{F}_1, \mathcal{F}_2, \ldots, \mathcal{F}_n$. \\
{\bf Output:} Intervals $I_i \in \mathcal{F}_i$, for $1 \leq i \leq n$, or \text{false}.
\end{tabularx}
  \begin{algorithmic}[1]
  		\STATE Set $P$ as the collection of all the endpoints of the given intervals, $P \coloneqq  \big\{p,q \in [0,1] \ \mid \  [p, q] \in \cup_{i=1}^n \mathcal{F}_i \big\} \ \cup \{0,1\}$. Index the points in $P = \{p_1, p_2, \ldots, p_{|P|}\}$ such that $0 = p_1 \leq p_2 \leq \ldots \leq p_{|P|} = 1$. \label{line:points-initialization}
		\STATE Initialize $T[q, i] = \text{false}$ for all $1 \leq q \leq |P|$ and all $0 \leq i \leq n$. \label{line:partition-set-initialization} \\
		\COMMENT{Table entry $T[q,i]$ will be \text{true} iff the interval $[p_1 =0, p_q]$ can be partitioned by some intervals $I_1 \in \mathcal{F}_1, I_2 \in \mathcal{F}_2, \ldots, I_i \in \mathcal{F}_i$.}
  		\STATE Set $T[1, 0] = \text{true}$. \COMMENT{Base case of the dynamic program.}
  		\FOR{$q = 1$ to $|P|$}
  			\FOR{$i= 1$ to $n$}
  				\STATE If there exists $1 \leq \widehat{q} < q$ such that $T[\widehat{q}, i-1] = \text{true}$ and interval $[p_{\widehat{q}}, p_{q}] \in \mathcal{F}_i$, then set table entry $T[q, i] = \text{true}$. \
  	    		\ENDFOR 
    	\ENDFOR 
	\STATE If $T[|P|,n] = \text{true}$, then retrace the \text{true} entries in the table and return the $n$ satisfying intervals $I_i \in \calF_i$. Otherwise, if $T[|P|,n] = \text{false}$, then return $\text{false}$. 
		\end{algorithmic}
\end{algorithm}

\begin{theorem}\label{theorem:partition-dp}
Given interval families $\mathcal{F}_1, \ldots, \mathcal{F}_n$ as input, \Cref{algo:interval-dp} (\textsc{IntervalSelect}) correctly determines in polynomial  time whether there exist $n$ pairwise disjoint intervals $I_i \in \calF_i$ that partition the cake and uphold the ordering (i.e., $I_i$ is to the left of $I_{i+1}$, for each $1 \leq i < n$). If such intervals exist, then the algorithm returns them. 
\end{theorem}
\begin{proof}
In \Cref{algo:interval-dp}, the table entry $T[q,i]$ is set to be \text{true} iff the interval $[p_1 =0, p_q]$ can be partitioned by some intervals $I_1 \in \mathcal{F}_1, I_2 \in \mathcal{F}_2, \ldots, I_i \in \mathcal{F}_i$. Here, the correctness of the algorithm directly follows an inductive argument over the implemented recurrence relation: $T[q,i]$ is set to be \text{true} iff $T[\widehat{q}, i-1] = \text{true}$ and interval $[p_{\widehat{q}}, p_{q}] \in \mathcal{F}_i$. Hence, the algorithm correctly identifies whether the desired intervals exist and returns them if they do. Considering two for-loops, we get that the time complexity of the algorithm is $O(n \ |P|^2) = O\left(n \ \left(\sum_{i=1}^n |\mathcal{F}_i| \right)^2 \right)$. This completes the proof of the theorem. 
\end{proof}

\subsection{Finding Approximately Equitable Cake Divisions}
\label{section:fptas-for-eq-div}
This section provides our approximation algorithm for equitable cake division. Here, we address cake valuations $f_i$ that are nonnegative and Lipschitz continuous. In particular, a cake valuation $f$ is said to be $\gamma$-Lipschitz continuous if, for all intervals $[\ell, r], [\ell', r'] \subseteq [0,1]$, we have 
\begin{align}
\left| f([\ell, r]) - f ([\ell', r']) \right| \leq \gamma \left\| \begin{pmatrix} \ell \\ r \end{pmatrix}  - \begin{pmatrix} \ell' \\ r' \end{pmatrix}  \right\|_1 \label{eqn:Lipschitz}
\end{align}
We assume that (an upper bound on) $\gamma$ is known for the given cake division instance. 

As mentioned, the ordering property obtained in Theorem \ref{theorem:eq-cake-nn} ensures that, under nonnegative valuations, we can restrict attention to equitable cake divisions $\calI^*=(I^*_1, \ldots, I^*_n)$ in which, for each $i \in [n]$, agent $i$'s interval, $I^*_i$, is the $i$th one from the left. \textsc{ApxEQ} (\Cref{algo:EQ-apx}) works towards finding such an ordered equitable division. The algorithm operates in the standard Robertson-Webb query model. In particular, the algorithm only requires (black-box) access to an oracle that, when queried with an interval $[p,q]$, returns the value $f_i([p,q])$, for the agents $i \in [n]$.  

The algorithm, for the given parameter $\varepsilon > 0$,  grid searches for a threshold $\tau$ that is $\varepsilon/4$ additively close to $\tau^* \coloneqq f_i(I^*_i)$. For each guessed $\tau$ and each agent $i \in [n]$, the algorithm constructs a family, $\calF^\tau_i$, that contains intervals $[p,q]$ of value $f_i([p,q]) \in \left[ \tau - \frac{\varepsilon}{2},  \tau + \frac{\varepsilon}{2}\right]$. The endpoints, $p$ and $q$, of the populated intervals are confined to be integer multiples of $\frac{\varepsilon}{8\gamma}$. This discretization bounds the cardinalities of the interval families $\calF^\tau_i$. 

Subsequently, the algorithm invokes, \textsc{IntervalSelect} (\Cref{algo:interval-dp}) as a subroutine to test whether there exist $n$ pairwise disjoint intervals $I_i \in \calF^\tau_i$ that partition the cake and uphold the ordering, i.e., $I_i$ is to the left of $I_{i+1}$ for each $1 \leq i < n$. We will show in the proof of Theorem \ref{theorem:fptas-for-equitable-cake-div} (stated below), that, for at least one $\tau$, \textsc{IntervalSelect} will succeed in finding the desired intervals $I_i$. We will also show that any such returned cake division $(I_1,\ldots, I_n)$ is approximately equitable. Theorem \ref{theorem:fptas-for-equitable-cake-div} states the guarantee for \textsc{ApxEQ} (\Cref{algo:EQ-apx}).
 
\floatname{algorithm}{Algorithm}
\begin{algorithm}[h]
\caption{\textsc{CakeApxEQ}} \label{algo:EQ-apx}
\begin{tabularx}{\textwidth}{X}
{\bf Input:} Cake division instance $\langle [n], \{f_i\}_{i=1}^n \rangle$, with $\gamma$-Lipschitz valuations, and parameter $\varepsilon \in \left(0, \frac{1}{4\gamma}\right)$. \\
{\bf Output:} Contiguous cake division $\mathcal{I} = (I_1, \ldots, I_n)$.
\end{tabularx}
  \begin{algorithmic}[1]
  		\STATE Define the set of grid points $G \coloneqq \{0, \frac{\varepsilon}{8 \gamma}, \frac{2\varepsilon}{8\gamma}, \frac{3 \varepsilon}{8\gamma}, \ldots, 1 \}$. Also, define collection $\calG$ of intervals whose endpoints are in $G$, i.e., $\calG  \coloneqq \left\{ [p,q] \mid p\leq q \text{ and } p,q \in G \right\}$. \label{line:eq-discretization} 
  		\FOR{each $\tau \in \left\{0, \frac{\varepsilon}{2}, \frac{2\varepsilon}{2}, \frac{3\varepsilon}{2}, \ldots, \gamma \right\}$} \label{line:eq-loop-start}
  			\STATE For each agent $i \in [n]$, define family $\calF^\tau_i \coloneqq \left\{ [p,q] \in \calG \ \mid \ f_i \left([p,q] \right) \in \left[\tau - \frac{\varepsilon}{2}, \tau + \frac{\varepsilon}{2} \right] \right\}$. \label{line:eq-family-def}
  			\STATE If \textsc{IntervalSelect}($\calF^\tau_1, \ldots, \calF^\tau_n$) successfully finds intervals $I_i \in \mathcal{F}^\tau_i$, then return cake division $\calI = (I_1, \ldots, I_n)$. Else, continue the for loop. \label{line:if-int-sel}
		\ENDFOR \label{line:eq-loop-end}
		\end{algorithmic}
\end{algorithm}
\begin{theorem}\label{theorem:fptas-for-equitable-cake-div}
Given any cake division instance, with $\langle [n], \{f_i\}_{i=1}^n \rangle$, with nonnegative and $\gamma$-Lipschitz continuous valuations, and parameter $\varepsilon \in \left(0, \frac{1}{4\gamma}\right)$ as input, \Cref{algo:EQ-apx} (\textsc{CakeApxEQ}) runs in time polynomial in $n$, $\gamma$, and $\frac{1}{\varepsilon}$, and it computes a contiguous cake division $\mathcal{I} = (I_1, \ldots, I_n)$ that satisfies $|f_i(I_i) - f_j(I_j)| \leq  \varepsilon$ for all $i,j \in [n]$. 
\end{theorem}
\begin{proof}
First, note that the for-loop (Line \ref{line:eq-loop-start}) in the algorithm iterates $O\left( \frac{\gamma}{\varepsilon} \right)$ times. Further, each iteration itself runs in polynomial time -- this follows from the $O\left( \left( \frac{\gamma}{\varepsilon}\right)^2 \right)$ size bound on the constructed interval families $\calF^\tau_i$ and efficiency of the subroutine \textsc{IntervalSelect} (Theorem \ref{theorem:partition-dp}). These observations imply that the algorithm runs in time polynomial in $n$, $\gamma$, and $\frac{1}{\varepsilon}$. 

We will prove that in Algorithm \ref{algo:EQ-apx}, for at least one value of $\tau$, the call to the \textsc{IntervalSelect} subroutine (in Line \ref{line:if-int-sel}) succeeds in finding the desired intervals. Further, we will show that the returned cake division $\mathcal{I}$ is approximately equitable, as stated.  

Since the given cake valuations $f_i$ are nonnegative and continuous, \Cref{theorem:eq-cake-nn} ensures the existence a contiguous equitable cake division $\mathcal{I}^* = (I^*_1, \ldots, I^*_n)$ in which, for all $i \in [n]$, the assigned interval $I^*_i$ is to the left of $I^*_{i+1}$ in the cake. Write $\tau^* \coloneqq f_i(I^*_i)$ and note that $\tau^* \leq  \gamma$. This bound follows from applying the Lipschitz continuity between interval $I^*_i=[e^*_{i-1}, e^*_i]$ and the empty interval $[1/2, 1/2]$:\footnote{Recall that $f_i(\emptyset) = 0$, by normalization.}
\begin{align*}
\tau^* = f_i(I^*_i) = \left| f_i(I^*_i) - f_i([1/2, \ 1/2]) \right| \leq \gamma \left\| (e^*_{i-1}, e^*_i) - (1/2, \ 1/2) \right\|_1 = \gamma \ \left( \left|e^*_{i-1} -1/2\right| + \left|e^*_i - 1/2 \right| \right) \leq  \gamma.
\end{align*} 
Hence, the for-loop (Line \ref{line:eq-loop-start}) of the algorithm selects, at some point during its execution, a $\tau$  that is additively $\frac{\varepsilon}{4}$ close to $\tau^*$. 

We will first show that for a $\tau$ that satisfies $| \tau - \tau^*| \leq \frac{\varepsilon}{4}$, the if-condition in Line \ref{line:if-int-sel} executes successfully. That is, for such a $\tau$, there exist intervals $I_i \in \calF^{\tau}_i$ that uphold the requirements of \Cref{theorem:partition-dp}. Hence, the subroutine \textsc{IntervalSelect} successfully returns intervals for the $\tau$ at hand. 

The desired intervals $I_i$ are obtained by rounding the endpoints of the intervals $I^*_i$ to lie within the set $G \coloneqq \{0, \frac{\varepsilon}{8 \gamma}, \frac{2\varepsilon}{8\gamma}, \frac{3 \varepsilon}{8\gamma}, \ldots, 1 \}$; this is the set of grid points defined in Line \ref{line:eq-discretization} of the algorithm. Write $0=e^*_0 \leq e^*_1 \leq e^*_2 \leq \ldots e^*_{n-1} \leq e^*_n = 1$ to denote all the endpoints of the intervals $I^*_i$. The ordering of these intervals imply that $I^*_i = [e^*_{i-1}, e^*_i]$ for each $i \in [n]$. For each index $j \in [n]$,  define $e_j$ to be the point in $G$ that is closest to $e^*_j$ from below, $e_j \coloneqq \max \{ x \in G \mid  x \leq e^*_j \}$. By definition, $e_j$s satisfy $0 = e_0 \leq e_1 \leq e_2 \leq \ldots \leq e_n = 1$. We define interval $I_i \coloneqq [e_{i-1}, e_i]$ for each $i \in [n]$. Note that the intervals $I_1, \ldots, I_n$ are pairwise disjoint, they partition the cake $[0,1]$, and are ordered left to right according to their index. In addition, since each $e_i \in G$, every interval $I_i \in \mathcal{G}$; where, $\mathcal{G}$ is the collection defined in Line \ref{line:eq-discretization} of the algorithm. The values of the intervals $I_i$ satisfy 
\begin{align}
|f_i(I_i) - f_i(I^*_i)| & \leq \gamma \ \|(e_{i-1}, e_{i}) - (e^*_{i-1}, e^*_i) \|_1 \tag{$f_i$ is $\gamma$ Lipschitz} \nonumber \\
& = \gamma \left( \frac{\varepsilon}{8 \gamma} +  \frac{\varepsilon}{8 \gamma} \right) \tag{by the granularity of $G$} \nonumber \\
& = \frac{\varepsilon}{4} \label{ineq:gran}
\end{align}
Equation (\ref{ineq:gran}) and the fact that $| \tau - \tau^*| \leq \frac{\varepsilon}{4}$ give us $f_i(I_i) \in \left[\tau - \frac{\varepsilon}{2}, \tau + \frac{\varepsilon}{2} \right]$ for each $i \in [n]$. Therefore, for the $\tau$ at hand and each $i \in [n]$, interval $I_i$ is included in the family $\calF^\tau_i$ in Line \ref{line:eq-family-def} of the algorithm. This inclusion and the properties of the intervals $I_i$ ensure that, for the families $\calF^\tau_1, \ldots, \calF^\tau_n$, the conditions in \Cref{theorem:partition-dp} hold and, hence, during this subroutine call \textsc{IntervalSelect} (\Cref{algo:interval-dp}) will return the desired intervals. Hence, \Cref{algo:EQ-apx} is guaranteed to return a contiguous cake division $\calI = (I_1, \ldots, I_n)$.

We will complete the proof by showing that the division $\calI=(I_1, \ldots, I_n)$ returned by \Cref{algo:EQ-apx} is approximately equitable. This directly follows from the fact that the intervals $I_i$ were obtained via subroutine \textsc{IntervalSelect} and, hence, for each $i \in [n]$, interval $I_i \in \calF^\tau_i$. The construction of the families $\calF^\tau_i$ (Line \ref{line:eq-family-def}) ensures that $f_i(I_i) \in \left[\tau - \frac{\varepsilon}{2}, \tau + \frac{\varepsilon}{2} \right]$, for each $i \in [n]$. Therefore, the returned division $\calI$ is approximately equitable: $\left| f_i(I_i) - f_j(I_j)\right| \leq \varepsilon$ for all $i, j \in [n]$. The theorem stands proved. 
\end{proof}

\subsection{Finding Approximately Equitable Allocations}\label{section:eq-three-nonnegative}
This section presents our algorithm (\Cref{algo:EQThree-identical}) for finding approximately equitable allocations of indivisible items under nonnegative valuations. We will prove that, given any discrete fair division instance $\langle [n], [m], \{v_i\}_{i=1}^n \rangle$ with nonnegative valuations as input,  \Cref{algo:EQThree-identical} computes an allocation $\alloc = (A_1, \ldots, A_n)$ that satisfies: $|v_i(A_i) -  v_j(A_j)| \leq 5 \Lambda + 1$. Here, $\Lambda$ is the maximum marginal value across the valuations, i.e., for each agent $i \in [n]$, all subsets $S \subseteq [m]$, and items $g \in [m]$, we have $|v_i(S \cup \{g\}) - v_i(S)| \leq \Lambda$. Furthermore, the returned bundles are guaranteed to be nonempty, $A_i \neq \emptyset$. This property of nonempty bundles has important implications for the graph-theoretic applications mentioned earlier; later in~\Cref{remark:nonempty-condition}, we show that our algorithm works even without the requirement of nonempty bundles, resulting in a slightly better fairness guarantee. Also, we will assume, without loss of generality, that $\Lambda \geq 1$. 

For the given discrete fair division $\calD$, \Cref{algo:EQThree-identical} (\textsc{AllocApxEQ}) starts by invoking the \textsc{CakeConstruction} method (Section \ref{section:cake-reduction}) to obtain a cake division instance, $\calC = \langle [n], \{f_i\}_{i=1}^n \rangle$. Then, to the obtained cake valuations $f_i$, the algorithm adds a piecewise linear function $\beta: \mathbb{R}_+ \mapsto \mathbb{R}_+$.\footnote{We will show that this padding implies that the eventually computed cake division will consist of intervals of sufficiently large length. This length property will extend to ensuring that, in the final allocation, each agent receives a nonempty bundle.} For each $x \in \mathbb{R}_+$, the function $\beta$ is defined as $\beta(x) \coloneqq \min \left\{ m \Lambda x, \ 4 \Lambda + 1 \right\}$. Specifically, the algorithm considers modified cake valuations $\widehat{f}_i$ defined as $\widehat{f}_i(J) \coloneqq f_i(J) + \beta(\len(J))$, for each interval $J \subseteq [0,1]$. Note that $\beta$ depends only on the length of the considered interval $J$. 

Next, the algorithm computes an approximately equitable cake division $\calI$ for the cake division instance $\widehat{\calC} \coloneqq \left\langle [n], \left\{ \widehat{f}_i \right\}_{i=1}^n \right\rangle$. Finally, using subroutine \textsc{CakeRounding} (Section \ref{section:CakeRounding}), the algorithm rounds $\calI$ to compute allocation $\calA$ for the given discrete fair division instance $\calD$. Note that while $\calI$ is computed as an approximately equitable cake division in $\widehat{\calC}$, the rounding is performed considering $\calC$. For one, \textsc{CakeRounding} requires it to receive a cake division instance that is constructed out of the underlying discrete instance (\Cref{lemma:cake-rounding-2}). Even with this distinction, we will prove that the returned allocation $\calA$ upholds the desired approximate equitability (Theorem \ref{theorem:EQ3ApxComputation}).   

We also note that all the subroutines invoked in \Cref{algo:EQThree-identical} only require query access to $f_i(J)$ and $\widehat{f}_i(J)$, for selected intervals $J$. Such queries are referred to as \text{Eval} queries in the standard Robertson-Webb query model. In the current context, \text{Eval} queries can be directly implemented by evaluating the given valuations $v_i$ over the indivisible items. Hence, to execute \Cref{algo:EQThree-identical} and its subroutines, it suffices to have value-oracle access to the valuations $v_i$s. 

\floatname{algorithm}{Algorithm}
\begin{algorithm}[h]
\caption{\textsc{AllocApxEQ}} \label{algo:EQThree-identical}
\begin{tabularx}{\textwidth}{X}
{\bf Input:} Discrete fair division instance $\calD = \langle [n], [m], \{v_i\}_{i=1}^n \rangle$ wherein the valuations are nonnegative and have marginal values at most $\Lambda$. \\
{\bf Output:} Allocation $\alloc = (A_1, \ldots, A_n)$.
\end{tabularx}
  \begin{algorithmic}[1]
		\STATE Set cake instance $\calC = \langle [n], \{f_i\}_{i=1}^n \rangle = \textsc{CakeConstruction}(\calD)$. \label{line:construct-f}
		\STATE Define function $\beta(x) \coloneqq \min \left\{ m \Lambda x, \ 4 \Lambda + 1 \right\}$ for each $x \in \mathbb{R}_+$. 
		\STATE For each agent $i \in [n]$ and each interval $J \subseteq [0,1]$, define cake valuation $\widehat{f}_i(J) \coloneqq f_i(J) + \beta \left( \len(J) \right)$. \label{line:define-hat-f}
		 \STATE With cake division instance $\widehat{\calC} \coloneqq \left\langle [n], \left\{ \widehat{f}_i \right\}_{i=1}^n \right\rangle$ and parameter $\varepsilon = \frac{1}{8 m \Lambda}$ as input, invoke \textsc{CakeApxEQ} (\Cref{algo:EQ-apx}) to compute cake division $\calI \coloneqq \textsc{CakeApxEQ} \left( \widehat{\calC} , \varepsilon \right)$. \label{line:call-cake-apx-eq}
		 \RETURN Allocation $\calA = \textsc{CakeRounding}(\calD, \calC, \calI)$. \label{line:call-round}
		\end{algorithmic}
\end{algorithm}

The theorem below provides our main result for \Cref{algo:EQThree-identical}. Before proving the theorem, we establish some useful lemmas. 

\begin{restatable}{theorem}{EQThreeApxComputation}\label{theorem:EQ3ApxComputation}
Let $\calD = \langle [n], [m], \{v_i\}_{i=1}^n \rangle$ be a discrete fair division instance in which the valuations are nonnegative and have marginal values at most $\Lambda$. Then, with $\calD$ as input, \Cref{algo:EQThree-identical} computes---in time polynomial in $n$, $m$ and $\Lambda$---an allocation $\alloc = (A_1, \ldots, A_n)$ with nonempty bundles ($A_i \neq \emptyset$) that satisfy $|v_i(A_i) - v_j(A_j)| \leq 5 \Lambda + 1$, for all $i, j \in [n]$.  
\end{restatable}

We will next bound the Lipschitz constant (see equation (\ref{eqn:Lipschitz})) of the constructed cake valuations $f_i$ and $\widehat{f}_i$.  
\begin{lemma}
\label{lemma:Lip-f}
The cake valuations $f_i$ (obtained in Line \ref{line:construct-f} of \Cref{algo:EQThree-identical}) are $\left( m \Lambda \right)$-Lipschitz continuous and nonnegative. 
\end{lemma}
\begin{proof}
Each valuation $f_i$ is defined as the composition $V_i$ and $b$ (see Section \ref{section:cake-reduction}). Here, $V_i$ is the multilinear extension of the given valuation $v_i$, and, for any interval $I \subseteq [0,1]$, the components of the function $b(I) = (b_1(I), \ldots, b_m(I)) \in [0,1]^m$ are defined as follows $b_k(I) = m \  \len \left(I \cap \left[\frac{k-1}{m}, \frac{k}{m} \right] \right)$. 

As mentioned previously, $f_i$s are continuous. Also note that the multilinear extensions, $V_i$s, are defined as an expected value with respect to nonnegative valuations $v_i$. Hence, each cake valuation $f_i = V_i \circ b$ is also nonnegative.  

We will next show that the Lipschitz constants of $V_i$ and $b$ are $\Lambda$ and $m$, respectively. Given that $f_i$ is the composition of these two functions, $f_i = V_i \circ b$, we obtain, as stated, that $f_i$ is $\left( m \Lambda \right)$-Lipschitz continuous.   

For the Lipschitz constant of $b$, consider any pair of intervals $[\ell,r], [\ell', r'] \subseteq [0,1]$. Now, a component-wise analysis gives us $\| b([\ell, r]) - b([\ell', r') \|_1 \leq m \left( |\ell - \ell'| + \ | r - r'| \right)$. Hence, the Lipschitz constant of $b$ is $m$. 

For the Lipschitz constant of $V_i$, consider any pair of vectors $x \in  [0,1]^m$ and $x + \eta e_g \in [0,1]^m$; here, $e_g \in [0,1]^m$ is vector with a $1$ in its $g$th component and zeros elsewhere. For any subset $S \subseteq [m] \setminus \{g\}$, write $p(S) \coloneqq  \prod_{j \in S} x_j \prod_{ j \in [m] \setminus \left(S + g \right)} (1 - x_j)$. That is, $p(S)$ denotes the probability of drawing $S$ from the product distribution induced by $(x_1, x_2, \ldots, x_{g-1}, x_{g+1}, \ldots, x_m)$.  Note that $\sum_{S \subseteq [m] \setminus \{g\}} \ p(S) = 1$. For the multilinear extension $V_i(x)$ we have 
\begin{align}
V_i(x) & = \sum_{T \subseteq [m]}  \left( \prod_{j \in T} x_j \allowbreak  \ \prod_{j \in [m] \setminus T} (1 - x_j) \right) v_i(T) = \sum_{S \subseteq [m] \setminus \{g\}} p(S) \Big(x_g v_i(S + g) \ + \  \left( 1-x_g \right) v_i(S) \Big) \label{eqn:split-mle}
\end{align} 
Similarly, $V_i(x + \eta e_g ) =  \sum_{S \subseteq [m] \setminus \{g\}} p(S) \Big( \left(x_g + \eta \right) v_i(S + g) \ + \  \left( 1-x_g - \eta \right) v_i(S) \Big)$. The difference of this equation and (\ref{eqn:split-mle}) gives us 
\begin{align}
V_i(x + \eta e_g) - V_i(x) & = \sum_{S \subseteq [m] \setminus \{g\}} p(S) \Big( \eta v_i(S + g) - \eta v_i(S) \Big) = \eta \sum_{S \subseteq [m] \setminus \{g\}} p(S) \Big( v_i(S + g) - v_i(S) \Big) \label{eq:mle-diff}
\end{align}
Using the triangle inequality and the bound $|v_i(S + g) - v_i(S)| \leq \Lambda$, from equation (\ref{eq:mle-diff}) we obtain 
\begin{align}
|V_i(x + \eta e_g) - V_i(x) | \leq \eta \sum_{S \subseteq [m] \setminus \{g\}} p(S) \left|  v_i(S + g) - v_i(S) \right|  \leq  \eta  \sum_{S \subseteq [m] \setminus \{g\}} p(S) \Lambda  = \eta \Lambda \label{ineq:mle-abs-diff}
\end{align} 
Here, the last equality follows from the fact that $\sum_{S \subseteq [m] \setminus \{g\}} \ p(S) = 1$. 

Therefore, for any two vectors $x, y \in [0,1]^m$, the triangle inequality and a component-by-component application of inequality (\ref{ineq:mle-abs-diff}) (with $\eta$s set as $|y_g - x_g|$s) give us 
\begin{align}
|V_i(y) - V_i(x)| \leq \Lambda \left( |y_1 - x_1| + \ldots + |y_m - x_m| \right) = \Lambda \| y - x  \|_1
\end{align}
Hence, $V_i$ is $\Lambda$-Lipschitz. As mentioned, this bound implies that the cake valuation $f_i = V_i \circ b$ is a $(m \Lambda)$-Lipschitz continuous function. The lemma stands proved. 
\end{proof}

Note that the piecewise linear function $\beta$ is $(m \Lambda)$-Lipschitz. Also, for any $x \geq 0$, the function value $\beta(x) \geq 0$, i.e., $\beta$ is nonnegative when applied to lengths of intervals. Hence, using the definition of the cake valuation $\widehat{f}_i(I) = f_i(I) + \beta(\len(I))$, we directly obtain the following corollary from Lemma \ref{lemma:Lip-f}.
\begin{corollary}\label{corollary:Lip-hat-f}
The cake valuations $\widehat{f}_i$ (obtained in Line \ref{line:define-hat-f} of \Cref{algo:EQThree-identical}) are $\left( 2 m \Lambda \right)$-Lipschitz continuous and nonnegative. 
\end{corollary}
 
The following lemma asserts that the algorithm---via the call to the \textsc{CakeApxEQ} (\Cref{algo:EQ-apx})---successfully finds a cake division $\calI$ in which the interval lengths are sufficiently large. 
\begin{lemma}
\label{lemma:call-to-cake-eq}
The call to \textsc{CakeApxEQ} (in Line \ref{line:call-cake-apx-eq}) successfully finds---in time polynomial in $n$, $m$, and $\Lambda$---a contiguous cake division $\calI=(I_1, \ldots, I_n)$ with the properties that 

\noindent
(i) $\left| \widehat{f}_i(I_i) - \widehat{f}_j (I_j) \right| \leq \frac{1}{8 m \Lambda}$, for all $i, j \in [n]$, and 

\noindent
(ii) $\len(I_i) \geq \frac{2}{m} + \frac{1}{2 \Lambda m^2}$, for all $i \in [n]$. 
\end{lemma} 
\begin{proof}
\textsc{CakeApxEQ} is invoked for cake division $\widehat{\calC} \coloneqq \left\langle [n], \left\{ \widehat{f}_i \right\}_{i=1}^n \right\rangle$. Note that, here, the cake valuations $\widehat{f}_i$s are normalized, nonnegative and $\left( 2 m \Lambda \right)$-Lipschitz continuous (Corollary \ref{corollary:Lip-hat-f}). Also, the selected parameter $\varepsilon = \frac{1}{8 m \Lambda}$. Hence, Theorem \ref{theorem:fptas-for-equitable-cake-div} holds and we obtain that the call to  \textsc{CakeApxEQ} runs in time polynomial in $n$, $m$, and $\Lambda$. In addition, Theorem \ref{theorem:fptas-for-equitable-cake-div} ensures that the returned contiguous cake division $\calI = (I_1, \ldots, I_n)$ satisfies $|\widehat{f}_i(I_i) - \widehat{f}_j(I_j) | \leq \varepsilon = \frac{1}{8 m \Lambda}$. This shows that property (i) in the lemma statement holds. 

For property (ii), assume, towards a contradiction, that, in the returned cake division $\calI$, there exists an interval $I_a$ with length $\len(I_a) < \frac{2}{m} + \frac{1}{2 \Lambda m^2}$. Note that the $n$ intervals in $\calI$ partition the cake $[0,1]$ and, hence, there exists an interval $I_b$ with length $\len(I_b) \geq  1/n$. We will show that the values of these intervals will differ by more than $\varepsilon = \frac{1}{8 m \Lambda}$; this fact would contradict property (i) established above. Towards the contradiction, first we upper bound the value of the interval $I_a$
\begin{align}
\widehat{f}_a (I_a) & = f_a(I_a) + \beta(\len(I_a)) \nonumber \\ 
& \leq m \Lambda \  \len(I_a) +  \beta(\len(I_a)) \tag{$f_a$ is $(m\Lambda)$-Lipschitz -- Lemma \ref{lemma:Lip-f}} \\
& < m \Lambda \left( \frac{2}{m} + \frac{1}{2 \Lambda m^2} \right) \ + \ \beta\left( \frac{2}{m} + \frac{1}{2 \Lambda m^2} \right) \tag{$\len(I_a) < \frac{2}{m} + \frac{1}{2 \Lambda m^2}$ and $\beta$ is non-decreasing} \\
& = 2 \Lambda + \frac{1}{2m} \ + \  \min \left\{ m \Lambda \left( \frac{2}{m} + \frac{1}{2 \Lambda m^2} \right), \ \ 4 \Lambda + 1 \right\} \nonumber \\
& = 4 \Lambda + \frac{1}{m} \label{ineq:val-Ia}
\end{align}
Next, we lower bound the value of the interval $I_b$. Here, we assume that $n \leq  \frac{m}{5}$.\footnote{In the complementary case, $n > m/5$, one can consider, for Theorem \ref{theorem:EQ3ApxComputation} directly, bundles $A_i$ of size at most five. This ensures $0 \leq v_i(A_i) \leq 5 \Lambda$ and, hence, the bound of $5 \Lambda + 1$, stated in the theorem, holds.}  
\begin{align}
\widehat{f}_b(I_b) & = f_b(I_b) + \beta(\len(I_b))  \nonumber \\
& \geq \beta\left( \len(I_b) \right) \tag{$f_b$ is nonnegative} \nonumber \\  
& \geq \beta \left( \frac{1}{n} \right) \tag{$\len(I_b) \geq  \frac{1}{n}$ and $\beta$ is non-decreasing} \\
& \geq \beta \left( \frac{5}{m} \right) \tag{$n \leq m/5$} \\
& = \min\left\{ m \Lambda \frac{5}{m}, \ \ 4 \Lambda + 1 \right\}  \nonumber \\
& = 4 \Lambda + 1 \label{ineq:val-Ib} 
\end{align} 
We can assume, without loss of generality, that $m \geq 2$. Hence, inequalities (\ref{ineq:val-Ia}) and (\ref{ineq:val-Ib}) reduce to $\widehat{f}_b(I_b) > \widehat{f}_a (I_a)  + \frac{1}{2}$. This bound, however, contradicts property (i) established above, since $\frac{1}{8 m \Lambda} < \frac{1}{2}$; recall that $m \geq 2$ and $\Lambda \geq 1$. Therefore, by way of contradiction, we obtain that, in the computed cake division $\calI=(I_1, \ldots, I_n)$, all the intervals have length $\len(I_i) \geq \frac{2}{m} + \frac{1}{2 \Lambda m^2}$. This establishes property (ii) and completes the proof of the lemma. 
\end{proof}

A useful implication of the length bound obtained in Lemma \ref{lemma:call-to-cake-eq} is that, in the rounded allocation $\calA$ (obtained in Line \ref{line:call-round}), the bundles are guaranteed to be nonempty. Formally, 
\begin{corollary}
\label{corollary:non-empty-bundles}
The allocation $\calA = (A_1, \ldots, A_n)$ (computed in Line \ref{line:call-round} via \textsc{CakeRounding}) consists of non-empty bundles: $A_i \neq \emptyset$ for all $i \in [n]$.  
\end{corollary}
\begin{proof}
In Line \ref{line:call-round}, the algorithm executes \textsc{CakeRounding} with the given discrete fair division instance $\calD$, the associated cake division instance $\calC$, and the cake division $\calI=(I_1, \ldots, I_n)$. For this division and all $i \in [n]$, we have $\len(I_i) \geq \frac{2}{m} + \frac{1}{2 \Lambda m^2} > \frac{2}{m}$ (Lemma \ref{lemma:call-to-cake-eq}). This length bound ensures that, for each $I_i$, there necessarily exists an index $k \in [m]$ such that the $k$th component of linear function $b$ (as defined in Section \ref{section:cake-reduction}) satisfies $b_k(I_i) = 1$. Hence, in Line \ref{line:define-B-i} of \textsc{CakeRounding} (\Cref{algo:cake-rounding}), each defined set $B_i \neq \emptyset$. Further, by construction (see Line \ref{line:table} of \textsc{CakeRounding}), we always have $B_i \subseteq A_i$. Therefore, the returned bundles $A_i$s are nonempty. This completes the proof of the corollary.  
\end{proof}
The last lemma in the current chain shows that the computed cake division, $\calI$, upholds approximate equitability even with respected to the cake valuations, $f_i$s, obtained directly from the given valuations, $v_i$s. 
\begin{lemma}
\label{lemma:apx-eq-wrt-f}
The computed cake division $\calI=(I_1, \ldots, I_n)$ satisfies $|f_i(I_i) - f_j(I_j)| \leq 2 \Lambda + 1$, for all $i, j \in [n]$.
\end{lemma}
\begin{proof}
Fix any pair of agents $i, j \in [n]$. By definition of the cake valuation $\widehat{f}_i$, we have 

\begingroup 
\allowdisplaybreaks
\begin{align}
f_i(I_i) & = \widehat{f}_i(I_i) - \beta( \len(I_i)) \nonumber \\
& \geq \widehat{f}_i(I_i) - (4 \Lambda + 1) \tag{by definition, $\beta(\cdot)$ is at most $4\Lambda + 1$} \\
& \geq \widehat{f}_j(I_j) - \frac{1}{8m\Lambda} - (4 \Lambda + 1) \tag{Lemma \ref{lemma:call-to-cake-eq}} \\
& = f_j(I_j) + \beta(\len(I_j)) - \frac{1}{8m\Lambda}  - (4 \Lambda + 1)   \nonumber \\
& \geq f_j(I_j) + \beta\left( \frac{2}{m} + \frac{1}{2 \Lambda m^2} \right)  - \frac{1}{8m\Lambda}  - (4 \Lambda + 1)  \tag{$\len(I_j) \geq \frac{2}{m} + \frac{1}{2 \Lambda m^2}$ - Lemma \ref{lemma:call-to-cake-eq}} \\
& = f_j(I_j) + 2 \Lambda + \frac{1}{2m} -  \frac{1}{8m\Lambda} - (4 \Lambda + 1) \nonumber \\
& \geq f_j(I_j) + 2 \Lambda - (4 \Lambda + 1) \tag{$\Lambda \geq 1$} \\
& = f_j(I_j) - (2 \Lambda + 1) \label{ineq:apx-ef-wrt-f}
\end{align}
\endgroup

Since inequality (\ref{ineq:apx-ef-wrt-f}) holds for all pairs $i, j \in [n]$, the lemma stands proved. 
\end{proof}
\subsubsection{Proof of Theorem \ref{theorem:EQ3ApxComputation}}
Lemma \ref{lemma:call-to-cake-eq} ensures that the algorithm successfully finds the contiguous cake division $\calI$ in Line \ref{line:call-cake-apx-eq}, and it does so in time polynomial in $n$, $m$, and $\Lambda$. In addition, \textsc{CakeRounding} (in Line \ref{line:call-round}) computes allocation $\calA$ in polynomial time. Hence, we get the stated time complexity for \Cref{algo:EQThree-identical}. Also, as desired, the bundles in the computed $\calA = (A_1, \ldots, A_n)$ are non-empty: $A_i \neq \emptyset$ (Corollary \ref{corollary:non-empty-bundles}). 

We next complete the proof by showing that $\calA$ is equitable within the stated additive margin. Towards this, note that the cake division $\calI=(I_1, \ldots, I_n)$ satisfies $f_i(I_i) \geq f_j(I_j) - (2 \Lambda + 1)$, for all $i, j \in [n]$ (Lemma \ref{lemma:apx-eq-wrt-f}). Hence, we can invoke Lemma \ref{lemma:cake-rounding-2} with $\alpha = 2 \Lambda + 1$ to obtain that, under $\calA = (A_1, \ldots, A_n)$ and for each $i, j \in [n]$, there exists subsets $A'_i, A'_j$ with the properties that $v_i(A'_i)  \geq v_j(A'_j) - (2 \Lambda + 1)$ and $| A_i \triangle A'_i| + |A_j \triangle A'_j| \leq 3$. In fact, the proof of the lemma gives us the following individual bounds: $| A_i \triangle A'_i| \leq 1$ and $|A_j \triangle A'_j| \leq 2$. Using these inequalities and the fact that $v_i$s have marginal values at most $\Lambda$, we obtain $v_i(A_i) \geq v_i(A'_i) - \Lambda$ along with $v_j(A'_j) \geq v_j(A_j) - 2 \Lambda$. We utilize these bounds to obtain
\begin{align}
v_i(A_i) & \geq v_i(A'_i) - \Lambda \nonumber \\
& \geq  v_j(A'_j) - (2 \Lambda + 1) - \Lambda \nonumber \\
& \geq v_j(A_j)  - 2 \Lambda - (3 \Lambda + 1) \nonumber \\
& = v_j(A_j) - (5 \Lambda + 1) \label{ineq:five-eq}
\end{align}
Since inequality (\ref{ineq:five-eq}) holds for all pairs $i, j \in [n]$, we get that the returned allocation $\calA$ satisfies equitability within an additive factor of $5 \Lambda + 1$. The theorem stands proved.

\begin{remark}\label{remark:nonempty-condition}
In Theorem \ref{theorem:EQ3ApxComputation}, if we relinquish the the requirement that the bundles have to be nonempty, then the following improved bound can be obtained: $|v_i(A_i) - v_j(A_j)| \leq 3 \Lambda + 1$. In particular, we can forgo adding the function $\beta$ and invoke the cake rounding directly with $f$. Applying the analysis directly on $f$ (instead of going via $\widehat{f}$) yields the stated improvement, though at the cost of  possibly finding allocations in which the bundles are empty. 
\end{remark}

\section{Finding Approximately Equitable Divisions Under Identical $\sigma$-Subadditive Valuations}
\label{section:find-apx-eq-id-sub}

This section presents our algorithmic results for instances wherein the agents' valuations are $\sigma$-subadditive and identical. Note that $c \sigma$-subadditive cake valuations can assign negative values to parts of the cake and, hence, we cannot apply Theorem \ref{theorem:eq-cake-nn} here. Nonetheless, since equitability and envy-freeness are equivalent criteria with identical valuations, in the current context of identical valuations, the envy-freeness guarantee of Theorem \ref{theorem:burnt-cake-envy-free} translates to an equitability one. Hence, we have the existence of a contiguous equitable cake division $\calI = (I_1, \ldots, I_n)$ in the current context.

In addition, since the agents' valuations are identical, we can assume, without loss of generality, that the intervals $I_i$s are ordered left-to-right on the cake in order of their index. These observations imply the following corollary of Theorem \ref{theorem:burnt-cake-envy-free}. 
\begin{corollary}
\label{corollary:eq-subadd-identical}
Let $\mathcal{C} = \langle [n], f \rangle$ be a cake division instance wherein the agents have an identical valuation $f$ that is continuous, $c\sigma$-subadditive, and globally nonnegative. Then, $\mathcal{C}$ admits a contiguous equitable cake division $(I_1, \ldots, I_n)$ in which interval $I_i$ is to the left of $I_{i+1}$, for all $1 \leq i < n$. 
\end{corollary}

Analogous to Theorem \ref{theorem:eq-cake-nn}, which held for nonnegative valuations, Corollary \ref{corollary:eq-subadd-identical} provides an existential and ordering property for identical $\sigma$-subadditive valuations. Hence, using a proof similar to the one of Theorem \ref{theorem:fptas-for-equitable-cake-div},\footnote{Recall that the the proof of Theorem \ref{theorem:fptas-for-equitable-cake-div} invokes Theorem \ref{theorem:eq-cake-nn}.} we obtain the following algorithmic result for equitable cake division under identical $\sigma$-subadditive valuations. 

\begin{theorem}\label{theorem:cake-alg-id-subadd}
Let $\calC = \langle [n], f \rangle$ be a cake division instance wherein the agents have an identical valuation $f$ that is $\gamma$-Lipschitz continuous, $c\sigma$-subadditive, and globally nonnegative. Then, given $\calC$ and parameter $\varepsilon \in \left(0, \frac{1}{4\gamma}\right)$ as input, we can compute---in time polynomial in $n$, $\gamma$, and $\frac{1}{\varepsilon}$---a contiguous cake division $\mathcal{I} = (I_1, \ldots, I_n)$ that satisfies $|f(I_i) - f(I_j)| \leq  \varepsilon$ for all $i,j \in [n]$. 
\end{theorem}

The cake division result in \Cref{theorem:cake-alg-id-subadd} can be bootstrapped to obtain an algorithmic result for finding an approximately equitable allocation of indivisible items under identical $\sigma$-subadditive valuations. Formally, 

\begin{theorem}\label{theorem:EQ3ApxComputation-subadditive}
Let $\calD = \langle [n], [m], v \rangle$ be a discrete fair division instance wherein the agents have an identical valuation $v$ that is $\sigma$-subadditive, normalized, and has marginal values at most $\Lambda$. Then, given $\calD$ as input, we can compute---in time polynomial in $n$, $m$ and $\Lambda$---an allocation $\alloc = (A_1, \ldots, A_n)$ with nonempty bundles ($A_i \neq \emptyset$) that satisfy $|v(A_i) - v(A_j)| \leq 5 \Lambda + 1$, for all $i, j \in [n]$.
\end{theorem}
\begin{proof}
The proof here follows along the same lines as in \Cref{theorem:EQ3ApxComputation}. In particular, we can utilize an algorithm similar to \Cref{algo:EQThree-identical}. Some of the technical details that differ are mentioned next. First, we note that for the given $\sigma$-subadditive $v$ the constructed cake valuation $f$ will be $c \sigma$-subadditive and globally nonnegative. Also, as was obtained in Section \ref{section:eq-three-nonnegative}, $f$ is $(m\Lambda)$-Lipschitz continuous. Therefore, we can bootstrap Theorem \ref{theorem:cake-alg-id-subadd} to obtain the desired result for approximate equitability in the indivisible-items setting as well. 
\end{proof}

\bibliographystyle{alpha}
\bibliography{references}

\end{document}